\documentclass[10pt,journal,compsoc]{IEEEtran}

\ifCLASSOPTIONcompsoc
  \usepackage[nocompress]{cite}
\else
  \usepackage{cite}
\fi

\usepackage{url}
\usepackage{amsmath, amsthm, amssymb}
\usepackage{balance}
\usepackage{graphicx}
\usepackage{epstopdf}
\usepackage{subcaption}
\usepackage{color}
\usepackage{indentfirst}
\usepackage{caption}
\usepackage{multirow}
\usepackage{array} 
\newcolumntype{L}[1]{>{\raggedright\let\newline\\\arraybackslash\hspace{0pt}}m{#1}} 
\newcolumntype{C}[1]{>{\centering\let\newline\\\arraybackslash\hspace{0pt}}m{#1}} 
\newcolumntype{R}[1]{>{\raggedleft\let\newline\\\arraybackslash\hspace{0pt}}m{#1}}

\usepackage{algorithm}
\usepackage{algpseudocode}
\usepackage{threeparttable}

\newtheorem{Theorem}{Theorem}[section]

\newtheorem{definition}{Definition}[section]

\newtheorem{remark}{Remark}[section]
\newtheorem{Claim}{Claim}[section]

\begin{document}
\title{Optimized Deployment of Autonomous Drones to Improve User Experience in Cellular Networks}
\author{Hailong~Huang,
       Andrey~V. Savkin,
       Ming~Ding,
       Mohamed Ali~Kaafar
\IEEEcompsocitemizethanks{\IEEEcompsocthanksitem{H. Huang is with School of Electrical Engineering and Telecommunications, University of New South Wales, Sydney 2052, Australia and Networks Research Group, DATA61,  CSIRO, Eveleigh, Sydney 2015, Australia. (E-mail: hailong.huang@unsw.edu.au).}
\IEEEcompsocthanksitem{A.V. Savkin is with School of Electrical Engineering and Telecommunications, University of New South Wales, Sydney 2052, Australia. (E-mail:  a.savkin@unsw.edu.au).}
\IEEEcompsocthanksitem{M. Ding is with Networks Research Group, Data61, CSIRO, Eveleigh, Sydney 2015, Australia (E-mail: Ming.Ding@data61.csiro.au).}
\IEEEcompsocthanksitem{M.A. Kaafar is with Networks Research Group, Data61, CSIRO, Eveleigh, Sydney 2015,  Australia (E-mail: Dali.Kaafar@data61.csiro.au).}
}}

\IEEEtitleabstractindextext{
\begin{abstract}
Modern wireless traffic demand pushes Internet Service Providers to develop effective strategies to improve user experience. Since deploying dense Base Stations (BSs) is not cost efficient, an alternative is to deploy autonomous drones to supplement existing BSs. A street graph is adopted to represent the area of interest. The outdoor User Equipments (UEs) to be served locate near streets and the 2D projections of drones are restricted to streets to avoid collision with buildings. We build up a UE density function based on a real dataset, reflecting the traffic in the area. We study four problems: where to deploy single drone to cover maximum UEs; where to deploy $k$ drones cover maximum UEs subject to an inner drone distance constraint; where to deploy $k$ drones cover maximum UEs subject to inner drone distance constraint and drones' battery constraints; and the minimum drones to cover a given percentage of UEs subject to inner drone distance constraint. We prove that the latter three problems are NP-hard and propose greedy algorithms with theoretical analysis. To our best knowledge, this is the first paper to consider the battery constraints for drone deployments. Extensive simulations have been conducted to verify the effectiveness of our approaches.

\end{abstract}
}
\maketitle
\IEEEdisplaynontitleabstractindextext
\IEEEpeerreviewmaketitle
\IEEEraisesectionheading{\section{Introduction}\label{introduction}}
\IEEEPARstart{E}{xplosive} demands for mobile data are driving mobile operators to respond to the challenging requirements of higher capacity and improved quality of user experience (QoE) \cite{nakamura2013trends}. Deploying more Base Stations (BSs) is able to meet the increasing traffic demand. This solution, however, may not only result in more cost for equipments and site rental, but also bring with other issues, such as a high percentage of BSs having low utility in non-peak hours. In this context, the utilization of autonomous drones, which work as flying BSs, could be a more efficient solution than network densification. 

Recently, the US government has approved a resolution to increase commercial use of drones\footnote{https://www.voanews.com/a/trump-ok-test-program-expand-domestic-drone-flights/4085752.html}. Arguably, drones will play a more significantly important role in our daily life in the future. In this paper, we focus on one of the key issues of content delivery using drones: drone deployment. Generally, the drone deployment problem has been studied to find out the optimal 3D positions for drones to serve user equipments (UEs) on a 2D plane. The altitudes of drones are restricted by local regulations. For example, in the US, the maximum allowable altitude is 120 meters above the ground \cite{H_usa}. In Australia, the drones should be more than 30 meters away from people \cite{H_australia}. 
Here, we consider a scenario where drones are flying at a fixed altitude within the allowed range. The drone deployment problem is a bit similar to the problem of optimal sensor placement in control \cite{savkin2001problem, savkin2002hybrid}. The difference is that in control the optimal sensor placement is done in the time domain, while the drone deployment considered in this paper is in the spatial domain.

Different from most existing work where UEs are assumed to be randomly distributed or following a predefined distribution in a 2D environment \cite{ding2016performance}, here we consider a more realistic scenario. We focus on urban environment, and only consider outdoor UEs. We propose a street graph and the UEs to be served by drones are near streets. This assumption is reasonable in urban environment as when users enter buildings, they can switch to Wi-Fi to access the network. Furthermore, instead of focusing on modelling the movements of UEs, we propose a UE density function model. Such model reflects the traffic demand at a certain position on the street graph during a certain time period, which can be obtained either by off-line learning the history traffic or by on-line crowdsensing \cite{crowdsensing11}. In this paper, we build up the UE density functions based on a realistic dataset collected from a social discovery mobile App: Momo\footnote{http://www.immomo.com}. We assume that the drones are wirelessly connected to the existing BSs \cite{backhaul2012} via high frequency radios, which do not interfere with the low frequency radios used by drones and UEs.

Commercial drones often rely on batteries to power their rotors and the on-board electronic modules, such as sensors and radios \cite{gupta2016survey}. Hence, the flying time allowed by the battery is limited. In this paper, we introduce the idea that the drones can recharge their battery from the existing powerlines\footnote{http://www.sbs.com.au/news/article/2017/08/23/powerlines-charge-drones-vic-students}. For safety, we assume that the drones can stop on the utility poles (instead of on powerlines) and recharge themselves. Then, the routine of a drone is serving UEs, flying to a utility pole, recharging on the utility pole, and then flying back to its serving position. According to the field experiments using Phantom drones \cite{azade2017}, the power for flying is over 140 watts; while the typical power for transmitting information through radios is usually around 250 milliwatts \cite{3GPP}, which is three orders of magnitude less than the former case. Compared to flying, the energy consumption caused by wireless transmission is neglected. To guarantee a certain time for serving UEs, the positions of drones should be well managed such that they are able to get recharged before running out of battery, since such positions impact the time spent on flying.

In this paper, we study four problems about drone deployment from simple to difficult:
\begin{itemize}
\item Single Drone Deployment (SDD): where to place a single drone in the area of interest to maximize the effectively served UE number;
\item $k$ Drones Deployment (kDD): given $k$ available drones, where to deploy them such that the effectively served UE number is maximized;
\item Energy aware $k$ Drones Deployment (EkDD): given $k$ available drones, the energy consumption and recharging models, where to deploy them such that the effectively served UE number is maximized;
\item Minimum Drones Deployment (MinDD): what is the minimum number of drones and where to deploy them such that a preferred UE coverage level can be achieved, subject to the inner drone distance constraint.
\end{itemize}

Clearly, the kDD problem is the general case of SDD by extending the number of drones to $k$ from 1, and the EkDD problem is the general case of kDD by taking into account the flying time of drones, which has not been considered in the context of using drones to serve UEs so far in the existing work. The minimum number of drones to achieve a certain user coverage level is a problem in which the Internet Service Providers (ISPs) are interested, which helps ISPs to consider the trade-off between the investment and benefit. Since we assume there a sufficient number of drones available, when drones are nearly out of battery, battery fresh drones would be deployed to replace them. Thus, the flying time constraint is not considered in MinDD.

We address SDD by finding the maximum coverage street point over the street graph, which is relatively easy. For kDD and EkDD, we prove that they can be reduced to the well known max $k$-cover problem, which is NP-hard \cite{feige1998threshold}. For MinDD, we prove that it can be reduced to the set cover problem, which is also NP-hard \cite{feige1998threshold}. Thus, we develop three greedy algorithms to solve kDD, EkDD and MinDD respectively.


To the best of our knowledge, this is the first work to study drone deployment constrained to street graph and with the consideration of battery lifetime constraint. We summarize our contributions as follows:
\begin{itemize}
\item We propose a street graph model describing the urban area of interest; and a UE density function reflecting the traffic demand at a certain position on street graph during a certain period of time.
\item We formulate a series of drone deployment problems and prove that they are NP-hard.
\item We design greedy algorithms to solve the proposed problems and provide theoretical analysis on the approximation factors.
\item Extensive simulations are conducted to verify the effectiveness of the proposed approaches.
\end{itemize}


The rest content is organized as follows. Section \ref{review} briefly reviews the related work. In Section \ref{model}, we present the system model, and in Section \ref{statement}, we formally formulate the problems. Section \ref{solution} presents the proposed solutions, which is followed by our theoretical analysis in Section \ref{analysis}. In Section \ref{simulation}, we conduct extensive simulations to evaluate the proposed approaches based on the realistic network dataset of Momo. Finally, Section \ref{conclusion} concludes the paper and discusses future work.
\section{Related Work}\label{review}
The use of autonomous drones in cellular networks has opened new paradigms in the area of networking. With the development of 5G network, drones will play a crucial role as a continuous network support \cite{sharma2017intelligent}.
 
There is a growing number of publications on the topic of drones in cellular network. Some publications study the fundamental wireless communication between drones and UEs. In \cite{al2014modeling}, air-to-ground pathloss is modelled. It shows that there are two main propagation situations: UEs have Line-of-Sight (LoS) with drones and UEs have non-Line-of-Sight (NLoS) with drones due to reflections and diffractions. A closed form expression for the probability of LoS between drones and UEs is developed in \cite{al2014optimal}. The authors of \cite{mozaffari2015drone} investigate the maximum coverage by two drones in the presence and absence of interference between them.

Other papers consider some interesting applications. In \cite{yang2017proactive}, a proactive drone deployment framework to alleviate the workload of existing ground BSs is proposed. It involves traffic models for three typical social activities: stadium; parade and gathering. Also, it presents a traffic prediction scheme based on the models. Further, it discusses an operation control method to evenly deploy the drones. 
\cite{becvar2017performance} studies the scenario where a drone serves users moving along a street. It compares with approaches of using static BSs and it shows that the drone introduces a significant gain in channel quality and outperforms that using dense BSs in terms of throughput.
\cite{chen2017caching} proposes a proactive deployment of cache-enabled drones to improve the quality of experience at users. The drones cache some popular content based on a prediction model. Such cache is able to reduce the data packet transmission delay. 
\cite{bor2016efficient} considers the problem of deploying one drone in a 3D environment to maximize the number of the served UEs. 
 Exploring the relationship between vertical and horizontal dimensions, \cite{alzenad20173d} solves the problem in \cite{bor2016efficient} by turning the problem into a circle placement problem. 
 In \cite{fotouhi2017dynamic}, the authors also treat the one-drone scenario and look into the energy efficiency issue for the downlink with consideration of both the mechanical energy consumption and the communication energy consumption. 
 \cite{sharma2016uav} studies the problem of deploying multiple drones. It develops an approach of mapping the drones to high traffic demand areas via a neural-based cost function. The authors of \cite{kalantari2016number} consider the issue of the number of drones and where to deploy them in 3D environment; and design a PSO based heuristic algorithm. 

In the mentioned work, the distribution of UEs is either random or following a predefined pattern over a 2D plane. However, such assumption may be far away from reality. In contrast, this paper focuses on the practical scenario, i.e., drones only serve outdoor UEs. We propose a street graph model and the UEs to be served are all near the streets. More importantly, the distribution of UEs is based on a real dataset Momo. Many existing approaches fail to consider the validity of drones' positions. We take into account this issue by restricting the 2D projections of drones to streets, which help avoid collision with tall buildings. None of the existing work has considered the lifetime of drones. The state of arts drones can only work for about half hours\footnote{https://www.dji.com}; while providing service to UEs should be a continuous operation. This paper novelly takes into account such issue. Next, we will present the system model and then formulate our drone deployment problems.
\section{System Model}\label{model}
We consider an urban area with UEs that cannot be served by the existing BSs due to capacity limitation or some malfunction of the infrastructure. We deploy drones to serve these UEs. The drones are assumed to be wirelessly connected to existing BSs via high frequency radios and there is a central station which manages the drones. 

In this section, we present the system model including a proposed street graph associated with the UE density function, and the wireless communication model. 

We first introduce the street graph model. Let ${\cal G} (V,E,\rho)$ be the street graph of the considered area, where $V$ and ${E}$ are the sets of street points and edges between two neighbour street points respectively. Each street point is associated with a UE density function $\rho_{v,t}$, $v\in V$, describing the number of UEs at $v$ during the $t^{th}$ ($t\in [1,T]$) time slot (the total time window we consider is $T$ time slots). Since $\rho$ is a function of time, the street graph $\cal G$ also varies with time. To simplify the statement, we omit the symbol $t$ in the below descriptions.

\begin{definition}\label{graph_distance}
The \textbf{graph distance} between two street points $v$ and $w$ on the graph, i.e., $g(v,w)$, is defined as the length of the shortest path between $v$ and $w$.  
\end{definition}

In the example shown in Fig. \ref{graph}, $g(A,B)$ is the length of the path represented by the red line segments; while $g(A,C)$ is the length of the path represented by the blue line segment. 

\begin{figure}[t]
\begin{center}
{\includegraphics[width=0.4\textwidth]{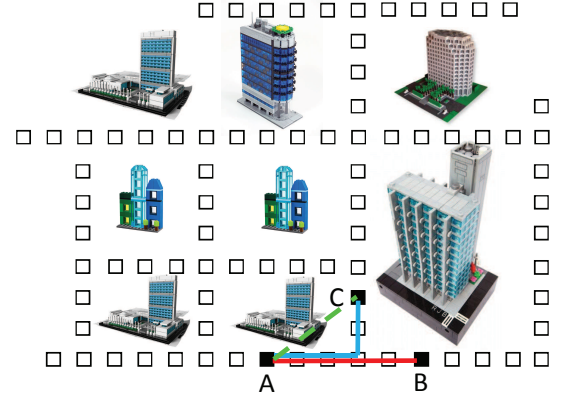}}
\caption{A street graph. The squares are street points. The length of the path represented by the red line segments is the graph distance between A and B; and that between A and C is the blue line.}\label{graph}
\end{center}
\end{figure}

\begin{remark}\label{remark1}
The physical ground distance between two street points on the street graph is always no larger than the corresponding graph distance, according to the triangle inequality theorem, which states that the sum of the lengths of two sides of a triangle must always be greater than the length of the third side.
\end{remark}

\begin{definition}\label{graph_distance}
Let $h$ be the altitude of the drone. The \textbf{spatial graph distance} between the drone at $(v,h)$ and a point $w$ on the graph, i.e., $d(v,h,w)$, is defined as:
\begin{equation}\label{physical_space_distance}
d(v,h,w)=\sqrt{g(v,w)^2+h^2}
\end{equation} 
\end{definition}

\begin{remark}\label{remark2}
According to Remark \ref{remark1}, $d(v,h,w)$ is always no smaller than the corresponding physical spatial distance.
\end{remark}

In this paper, we do not focus on optimizing the altitude of the drones\footnote{This may involve dealing with specific regulations and policies.}. For simplicity, we fix $h$ in this paper and then the $h$ in the spatial graph distance can be omitted without confusion. To avoid hitting tall buildings, we assume the drones can only fly and hover head over the streets.

Now, we consider the wireless communication model between drones and UEs, i.e., the drones-to-UEs links rather than the BSs-to-drones wireless backhaul links. The drones may have LoS or NLoS with UEs depending on the terrain. Consider a drone hovering over street point $v$ and a UE at street point $w$. Adopting the realistic 3GPP propagation model \cite{3GPP}, we have the path loss of signal from the drone to the UE: 
\begin{equation}\label{pathloss_simple}
\begin{aligned}
PL^{\lambda}(d(v,w))=A^{\lambda}+B^{\lambda}\log(d(v,w))\\
\end{aligned}
\end{equation}
where $\lambda \in \{LoS, NLoS\}$. 

The received power by the UE is computed by:
\begin{equation}\label{received_power}
S(v,w)=\frac{P_{tx}}{PL^{\lambda}(d(v,w))}
\end{equation}
where $P_{tx}$ is the transmit power of the drone. 

Then, the signal to noise ratio (SNR) is:
\begin{equation}\label{SNR}
SNR(v,w)=\frac{S(v,w)}{N_0}
\end{equation}
where $N_0$ is the power of noise. 

To meet the Quality of Service (QoS) requirement, we require that the SNR at any UEs should not be lower than a given threshold $\alpha$:
\begin{equation}\label{SNR_alpha}
SNR(v,w) \geq \alpha
\end{equation} 
otherwise, the received signal cannot be demodulated by the normal UEs.

From (\ref{pathloss_simple}), (\ref{received_power}) and (\ref{SNR}), we find that whether (\ref{SNR_alpha}) can be satisfied depends on $d(v,w)$ ($g(v,w)$, as $h$ is fixed in this paper). Moreover, we find that the graph distance is upper bounded, beyond which (\ref{SNR_alpha}) will never be met. Let $g_{max}$ denote such upper bound.
\begin{remark}\label{remark3}
As will be shown later, we use the NLoS situation to compute $g_{max}$, which is the worst case. Then, according to Remark \ref{remark2}, if a drone is placed at $v$, the UEs within graph distance $g_{max}$ on the street graph always satisfy (\ref{SNR_alpha}), i.e., can be served by the drone.
\end{remark}

When multiple drones are used, the signal to interference and noise ratio (SINR) at a UE is computed by:
\begin{equation}\label{SINR}
SINR(v,w)=\frac{S(v,w)}{I+N_0}
\end{equation}
where $I$ is the total interference from all the other drones.
\section{Problem Statement}\label{statement}
In this section, we formulate four drone deployment problems.
\subsection{Single Drone Deployment (SDD)}
We start from the simplest problem: Single Drone Deployment (SDD) i.e., where to deploy a single drone in the area of interest such that the effectively served UE number by the drone is maximized.

Before formulating the problem, we introduce the below definitions.

\begin{definition}\label{covered_set}
The \textbf{covered street point set} is defined as the set of street points, which are covered by the drone. If the drone is deployed at $v$, the covered street point set can be computed by: $S(v)=\{w|g(w,v)\leq g_{max}\}$. Correspondingly, we define the \textbf{covered UE set} as the set of UEs which can be served by the drone if it is placed at $v$, i.e., $U(v)$.
\end{definition}

\begin{remark}\label{remark4}
Let $\cal U$ denote the set of all UEs on the street graph, then we have ${\cal U}=\bigcup_{v\in V} U(v)$.
\end{remark}

\begin{definition}\label{benefit}
The \textbf{benefit} achieved by placing a drone at $v$ is defined as the number of UEs it can serve, i.e., $B(v)=\sum_{w\in S(v)} \rho_w$.
\end{definition}

\begin{remark}\label{remark5}
From Definition \ref{covered_set}, \ref{benefit}, we have $|U(v)|=B(v)$, $\forall\ v\in V$.
\end{remark}

We introduce a binary variable $z(v)$ indicating whether the drone is deployed at $v$:
\begin{equation}\label{z}
z(v)=
\begin{cases}
\begin{aligned}
&1, \text{if the drone is deployed at}\ v\\
&0, \text{otherwise}
\end{aligned}
\end{cases}
\end{equation}

With Definition \ref{covered_set}, \ref{benefit} and (\ref{z}), Single Drone Deployment (SDD) can be formulated as follows:
\begin{equation}\label{optimization1}
\max_{z(v)} \left( \sum_{w\in S(v)} \rho_w \right) z(v)
\end{equation}
subject to 
\begin{equation}\label{constraint1}
S(v)=\{w|g(w,v)\leq g_{max}\}
\end{equation}

\subsection{$k$ Drones Deployment (kDD)}
Now we consider the second problem: $k$ Drones Deployment (kDD), i.e., where to deploy $k$ drones such that the effectively served UE number is maximized subject to the minimum distance between any to drones. Another constraint, i.e., the battery constraint, will be considered in Section \ref{P4}.

In the scenario with multiple drones, a UE may hear signals from more than one of them, which causes interference. In this paper, we require that the minimum distance between any two drones is $\beta$. Note that, the system parameter $\beta$ controls the intensity of interference and a lower $\beta$ may lead to higher interference.

We formulate the $k$ Drones Deployment (kDD) problem as follows:
\begin{equation}\label{optimization2}
\max_{z(v)} \sum_{v\in V}\left( \sum_{w\in S(v)} \rho_w \right) z(v)
\end{equation}
subject to
\begin{equation}\label{constraint2_1}
S(v)=\{w|g(w,v)\leq g_{max}\}
\end{equation}
\begin{equation}\label{constraint2_2}
\sum_{v\in V} z(v)=k
\end{equation}
\begin{equation}\label{constraint2_3}
g(v,w)>\beta,\ \forall\ z(v)=1, z(w)=1, v\neq w
\end{equation}

As seen in the objective function (\ref{optimization2}), we want to maximize the benefits of drone deployment. 
In comparison to the SDD problem, here (\ref{constraint2_3}) constrains the inner distance between any pair of drones.

\subsection{Energy aware $k$ Drones Deployment (EkDD)}\label{P4}
The third problem is named as Energy aware $k$ Drones Deployment (EkDD), which aims at deploying $k$ drones such that the effectively served UE number is maximized, subject to the inner drone distance constraint and the drone battery lifetime constraint. 

The off-the-shelf drones often rely on the preloaded battery to power their rotors and the on-board electronics models. Considering the battery lifetime constraint, we require an efficient duty management scheme for the available drones to guarantee their service. 

Typically, there are three states of a drone: Serve (S), Fly (F) and Recharge (R). Denote $\lambda_S$, $\lambda_F$, and $\lambda_R$ as the percentages of a time slot during which a drone is in S, F, and R respectively. It is apparent that, $\lambda_S+\lambda_F+\lambda_R=100\%$. Among these three states, we assume that only S contributes to network performance, because when a drone is in F or R, the wireless backhaul links may be unstable or unavailable.

We assume that there are a few street points on the graph, which allow the drones to recharge. In practice, such positions can be placed on the utility poles. Let $V_R \subset V$ be the set of recharging positions.

Let $s$ be the speed of the drones. Given $\lambda_S$, $\lambda_F$, and $\lambda_R$, we can compute a maximum graph distance $g_R$ from the recharging positions, beyond which a drone cannot guarantee its serving time, i.e., $\lambda_S$ percent of the full cycle of time slots (S $\rightarrow$ F $\rightarrow$ R $\rightarrow$ F). For simplicity, we assume that to reach a recharging position $v_R$ from its current serving position $v$, a drone first flies over the street to the position whose projection is $v_R$, it then lands on the utility pole (whose altitude is $h_R$) to recharge. When the recharging process finishes, it ascends to $h$ and then flies to the serving position $v$ following streets. Consequently, we have
\begin{equation}
2(g(v,w)+h-h_R) \leq s\lambda_F\cdot 1,
\end{equation}
from which we can obtain
\begin{equation}
g_R=\frac{1}{2}s\lambda_F+h_R-h.
\end{equation}

Now, we are ready to formulate the problem with the constraint of recharging positions, which is shown below.

\begin{equation}\label{optimization4}
\max_{z(v)} \lambda_S\sum_{v\in V}\left( \sum_{w\in S(v)} \rho_w \right) z(v)
\end{equation}
subject to
\begin{equation}\label{constraint4_1}
S(v)=\{w|g(w,v)\leq g_{max}\}
\end{equation}
\begin{equation}\label{constraint4_2}
\sum_{v\in V} z(v)=k
\end{equation}
\begin{equation}\label{constraint4_3}
g(v,w)>\beta,\ \forall\ z(v)=1, z(w)=1, v\neq w
\end{equation}
\begin{equation}\label{constraint4_4}
g(v,v_R)\leq g_R,\ \forall\ z(v)=1,\ \exists\ v_R \in V_R 
\end{equation}

Here, we still aim at finding the optimal positions for $k$ drones such that the total number of served UEs can be maximized and in the meanwhile any two drones are at least $\beta$ apart from each other and any drone must be within $g_R$ from the nearest recharging position.

\subsection{Minimum Drones Deployment (MinDD)}
The above three problems all consider the scenario when drones are limited. The final problem is to seek the minimum number of drones and their deployments such that an ISP preferred level of coverage of the UEs can be served (say 98\%), which is called Minimum Drones Deployment (MinDD). Since we assume that we have a sufficient number of drones, the battery constraint can be ignored but the inner drone distance constrained is still considered.

MinDD can be formulated as follows:

\begin{equation}\label{optimization3}
\min_{z(v)} \sum_{v\in V} z(v)
\end{equation}
subject to
\begin{equation}\label{constraint3_1}
|\bigcup_{\{v|z(v)=1\}}U(v)|\geq\gamma |{\cal U}|
\end{equation}
\begin{equation}\label{constraint3_2}
g(v,w)>\beta,\ \forall\ z(v)=1, z(w)=1, v\neq w
\end{equation}
where $\gamma$ in (\ref{constraint3_1}) is the ISP preferred level of coverage, i.e., the effectively served UE ratio should not be less than $\gamma$. 


\section{Proposed solutions}\label{solution}
In this section, we present the solutions to the four considered problems.
\subsection{Solution to SDD}
\label{solution_single_case}
It is relatively easy to solve SDD using the maximum coverage street point over the street graph.

To find the solution, we can search the set of $V$. For each $v\in V$, find $S(v)$ and compute $B(v)$. Then, the $v$ having the largest $B(v)$ is the optimal place. Note, the solution to the problem may not be unique due to the discrete nature of the UE number. 

\subsection{Solution to kDD and EkDD}\label{solution_multiple_case}
First, we characterize the kDD and EkDD problems as NP-hard ones in Theorem \ref{theorem1}.

\begin{Theorem}\label{theorem1}
The problems of kDD and EkDD are NP-hard.
\end{Theorem}
\begin{proof}
If we ignore constraint (\ref{constraint2_3}) in kDD, and constraint (\ref{constraint4_3}) and (\ref{constraint4_4}) in EkDD, the two problems are reduced to an instance of max $k$-cover problem, i.e., given a universe $\cal U$ and sets $U(v)\in {\cal U}$ ($v\in V$), we are looking for a set ${\cal V}$ consisting of $k$ elements of $v$, such that the union of $U(v)$ has maximum cardinality. As shown in \cite{feige1998threshold}, the max $k$-cover problem is NP-hard. So, both of kDD and EkDD are also NP-hard.
\end{proof}

We propose greedy algorithms to solve kDD and EkDD respectively as shown in Algorithms \ref{algorithm1} and \ref{algorithm2}. The basic idea of Algorithm \ref{algorithm1} follows that in \cite{korte2012combinatorial}, i.e, in each iteration, we pick a street point which 1) adds a maximum benefit to the current benefit. Since we consider the inner drone distance constraint, we have to further check 2) whether the street point is $\beta$ away from all the already picked street points.

Clearly, the set ${\cal V}=\{v|z(v)=1\}$, i.e., the set of street points to deploy drones, is with a cardinality of $|{\cal V}|=k$. 
Let $C$ be the set of UEs which have already been served by drones. Initially, both $\cal V$ and $C$ are empty sets. In each iteration, we first compute the optimal street point $v$ in $V$. Second, we check whether (\ref{constraint2_3}) in kDD is satisfied. If yes, we add the served UEs by $v$ to the set of $C$ and add $v$ to the set of $\cal V$. Otherwise, we remove $v$ from $V$.

The intuition of Algorithm \ref{algorithm2} is similar to Algorithm \ref{algorithm1}. For EkDD,  with regards to $\beta$ constraint, we have another condition to check, i.e., 3) whether the selected street point is within $g_R$ from the nearest recharging position, see Line 7 of Algorithm \ref{algorithm2}. 

\begin{algorithm}[t]
\caption{The greedy algorithm for kDD}\label{algorithm1}
\begin{algorithmic}[1]
\Require $V$, $U$, $\beta$ and $k$
\Ensure $\cal V$
\State $C=\emptyset$.
\State ${\cal V}=\emptyset$.
\State $i=1$.
\While{$i \leq k$}
\State $\max_{v\in \{v|v\in V, v\notin {\cal V}\}} |C\bigcup U(v)|$.
\If {$g(v,w)> \beta$, $\forall\ w\in {\cal V}$}
\State $C=C\bigcup U(v)$.
\State ${\cal V}={\cal V}\bigcup \{v\}$.
\EndIf
\State $V=V\setminus v$.
\State $i=i+1$.
\EndWhile
\end{algorithmic}
\end{algorithm}

\begin{algorithm}[t]
\caption{The greedy algorithm for EkDD}\label{algorithm2}
\begin{algorithmic}[1]
\Require $V$, $U$, $\beta$, $k$, $V_R$ and $g_R$
\Ensure $\cal V$
\State $C=\emptyset$.
\State ${\cal V}=\emptyset$.
\State $i=1$.
\While{$i \leq k$}
\State $\max_{v\in \{v|v\in V, v\notin {\cal V}\}} |C\bigcup U(v)|$.
\If {$g(v,w)> \beta$, $\forall\ w\in {\cal V}$}
\If {$g(v,v_R)\leq g_R$, $\exists\ v_R \in V_R$}
\State $C=C\bigcup U(v)$.
\State ${\cal V}={\cal V}\bigcup \{v\}$.
\EndIf
\EndIf
\State $V=V\setminus v$.
\State $i=i+1$.
\EndWhile
\end{algorithmic}
\end{algorithm}

One interesting point worth mentioning is that not all of these $k$ drones can serve simultaneously if the battery constraint is considered, instead, they need to fly to the recharging positions alternately. Let $n$ be the number of drones which must recharge to make the drones work sustainablely. Then, we have that the amount of consumed energy should be no larger than the recharged energy, i.e.,
\begin{equation}\label{recharge_number}
p(k-n) \leq qn
\end{equation}
where $p$ is the energy consumption per time slot and $q$ is the recharging energy per time slot. Here, we only consider the dominant part of the power usage for resisting the gravity and omit that for other marginal parts of energy consumption such as wind. As a result, flying and hovering of drones consume the same power.

From (\ref{recharge_number}), we can obtain:
\begin{equation}\label{n_k}
n \geq \frac{q}{q+p} k
\end{equation}
Then, the drones can be divided into $\lfloor\frac{q+p}{q}\rfloor$ groups to recharge by turns. 

\subsection{Solution to MinDD}\label{minimum_number}
We characterize the MinDD problem as a NP-hard one in Theorem \ref{theorem2}.
\begin{Theorem}\label{theorem2}
The problem of MinDP is NP-hard.
\end{Theorem}
\begin{proof}
If we ignore constraint (\ref{constraint3_2}) and set $\gamma=100\%$ in MinDD, the problem is reduced to an instance of set cover problem, i.e., given a universe $\cal U$ and sets $U(v)\in {\cal U}$ ($v\in V$), we are looking for a collection ${\cal V}$ of the minimum number of sets from $U$, whose union is the entire universe $\cal U$. Formally, $\cal V$ is a set cover if $\bigcup_{v\in {\cal V}}U(v)={\cal U}$. We try to minimize $|{\cal V}|$. As shown in \cite{feige1998threshold}, the set cover problem is NP-hard, then the generalized and constrained version of set cover problem, i.e., MinDD, is also NP-hard.
\end{proof}

We design a greedy algorithm for MinDD, which is shown in Algorithm \ref{algorithm3}. The basic idea of Algorithm \ref{algorithm3} is similar to Algorithms \ref{algorithm1}, i.e., in each iteration, it picks the street point which leads to the maximum coverage of UEs and at the same time satisfying the inner drone distance constraint.  Algorithm \ref{algorithm3} will terminate when $\gamma$ percent of all the UEs in $\cal U$ have been covered by the drones. The output $\cal V$ gives the projections of drones and the number of drones is $|\cal V|$.


\begin{algorithm}[t]
\caption{The greedy algorithm for MinDD}\label{algorithm3}
\begin{algorithmic}[1]
\Require $V$, $U$, $\cal U$ and $\beta$
\Ensure ${\cal V}$
\State $C=\emptyset$.
\State ${\cal V}=\emptyset$.
\While{$|C| < \gamma |{\cal U}|$}
\State $\max_{v\in \{v|v\in V, v\notin {\cal V}\}} |C\bigcup U(v)|$.
\If {$g(v,w)> \beta$, $\forall\ w\in {\cal V}$}
\State $C=C\bigcup U(v)$.
\State ${\cal V}={\cal V}\bigcup \{v\}$.
\EndIf
\State $V=V\setminus v$.
\EndWhile
\end{algorithmic}
\end{algorithm}

\section{Analysis of the proposed algorithms} \label{analysis}
This section analyses the proposed greedy algorithms.

We first consider Algorithms \ref{algorithm1} and \ref{algorithm2}.
Let $OPT$ be the number of served UEs by an optimal solution. Let $c_i$ be the number of served UEs by the $i^{th}$ ($i\leq k$) picked street point, and $x_i$ be the number of served UEs by $i$ already picked street points, i.e., $x_i=\sum_{j=1}^i c_j$. Then, the remaining unserved UE number can be computed by $y_i=OPT-x_i$. Further, $x_0=0$, then $y_0=OPT$.

Since the optimal solution uses $k$ sets to serve $OPT$ UEs, we find that: in the $(i+1)^{th}$ step, some street point must be able to serve at least $\frac{1}{k}$ of the remaining uncovered UEs from $OPT$, i.e., 
\begin{equation}\label{ana1}
x_{i+1}\geq \frac{y_i}{k}
\end{equation}
In particular, when $i=0$, we have:
\begin{equation}\label{ana2}
x_1\geq \frac{OPT}{k},
\end{equation}
since $y_0=OPT$.

\begin{Claim}\label{claim1}
$y_{i+1}\leq (1-\frac{1}{k})^{i+1}OPT$.
\end{Claim}
\begin{proof}
We use the method of Mathematical Induction to prove the claim.

When $i=0$, from (\ref{ana2}) we have
\begin{equation}\label{ana4}
\begin{aligned}
&x_1 \geq \frac{OPT}{k}\\
\Rightarrow &OPT-x_1 \leq OPT-\frac{OPT}{k}\\
\Rightarrow &y_1 \leq (1-\frac{1}{k}) OPT
\end{aligned}
\end{equation}
i.e., the claim is true when $i$ takes 0.

Further, we assume $y_{i}\leq (1-\frac{1}{k})^{i}OPT$ and then we derive $y_{i+1}\leq (1-\frac{1}{k})^{i+1}OPT$ below:
\begin{equation}\label{ana3}
\begin{aligned}
&x_{i+1}\leq y_{i}-y_{i+1}\\
\Rightarrow &y_{i+1}\leq y_{i}-x_{i+1}\\
\Rightarrow &y_{i+1}\leq y_{i}-\frac{y_i}{k}\ (\text{using}\ (\ref{ana1}))\\
\Rightarrow &y_{i+1}\leq (1-\frac{1}{k})^{i+1}OPT.
\end{aligned}
\end{equation}

Therefore, (\ref{ana4}) and (\ref{ana3}) prove Claim \ref{claim1}.
\end{proof}

Now we are in the position to present the main results. 

\begin{Theorem}\label{theorem3}
Algorithm \ref{algorithm1} (Algorithm \ref{algorithm2}) is a $(1-\frac{1}{e})$ approximation of kDD (EkDD).
\end{Theorem}

\begin{proof}
From Claim \ref{claim1}, we have
\begin{equation}
\begin{aligned}
y_{k} &\leq (1-\frac{1}{k})^{k}OPT\\
&\leq \frac{1}{e}OPT
\end{aligned}
\end{equation}
Hence, we have
\begin{equation}
x_k=OPT-y_k\geq (1-\frac{1}{e})OPT
\end{equation}
i.e., the number of served UEs by the $k$ street points picked by Algorithm \ref{algorithm1} (Algorithm \ref{algorithm2}) is at least $(1-\frac{1}{e})$ fraction of the optimal solution $OPT$.
\end{proof}

The above results can also be used in the analysis of Algorithm \ref{algorithm3}. Let $k^*$ denote the optimal solution to MinDD, i.e., $k^*$ is the minimum number of required drones to serve at least $\gamma$ of all the UEs. Let $m=\gamma |{\cal U}|$. Clearly, if we set $k=k^*$ in kDD or EkDD, $OPT=m$. 
\begin{Theorem}\label{theorem4}
Algorithm \ref{algorithm3} is a $(\ln (m)+1)$ approximation of MinDD.
\end{Theorem}
\begin{proof}
From Claim \ref{claim1}, we have $y_i\leq (1-\frac{1}{k^*})^i m$. After picking $k^* \ln (\frac{m}{k^*})$ street points, the remaining unserved UE number is:
\begin{equation}
\begin{aligned}
y_i&\leq (1-\frac{1}{k^*})^{k^* \ln (\frac{m}{k^*})} m\\
&\leq \frac{1}{e}^{\ln (\frac{n}{k^*})} m \\
&=k^*\\
\end{aligned}
\end{equation}

The worst case of serving those remaining unserved UEs (at most $k^*$) is to use at most $k^*$ drones, i.e., each drone serves one UE. Then, we have the minimum number of drones obtained by Algorithm \ref{algorithm3} is
\begin{equation}
\begin{aligned}
k^* \ln (\frac{m}{k^*})+k^* &\leq k^* (\ln (\frac{m}{k^*})+1)\\
&\leq k^* (\ln (m)+1)\\
\end{aligned}
\end{equation}

Therefore, the minimum number of drones to serve at least $\gamma$ of all the UEs obtained by Algorithm \ref{algorithm3} is a $(\ln(m)+1)$ approximation of the optimal solution.
\end{proof}

This section analyses the approximation factors for the proposed greedy algorithms. In the next section, we will evaluation these algorithms through extensive simulations.
\section{Evaluation}\label{simulation}
In this section, we evaluate our proposed drone deployment strategies based on a collected dataset from a mobile App (Momo). We first introduce the parameters used in the simulation (Section \ref{setup}). Then, we provide details on the realistic dataset in use (Section \ref{dataset}). Next, we show the metrics for evaluation (Section \ref{metric}), followed by the compared approaches (Section \ref{compared_approach}). We present the extensive simulation results are shown (Section \ref{results}). 
\subsection{Simulation Setup}\label{setup}
Table \ref{table_parameter} provides a quick reference for the used parameters in the simulation.

According to the QoS constraint, we can determine $g_{max}$ from (\ref{physical_space_distance}), (\ref{pathloss_simple}), (\ref{received_power}), (\ref{SNR}) and (\ref{SNR_alpha}). We vary $\alpha$ from 10 to 20 $dB$ and calculate the corresponding $g_{max}$ for both LoS and NLoS cases, as shown in Fig. \ref{G_alpha}. We can find that under the same $\alpha$, the LoS case has a larger $g_{max}$ than of NLoS case. To determine the value of $g_{max}$, we further show the probability of LoS under varying distance between transmitter and receiver in Fig. \ref{p_los}, using the model obtained by realistic experiments in \cite{3GPP}. From Fig. \ref{p_los}, we can see that when the distance between the transmitter and receiver is greater than 100 meters, the probability of LoS is less than 20\%. Thereby, we use the NLoS case to compute $g_{max}$. We select $\alpha=15$ $dB$ and the corresponding $g_{max}$ becomes 95 $m$.

\begin{table}[t]
\begin{center}
\caption{Parameter configuration}\label{table_parameter} 
  \begin{tabular}{| l | l | l |}
    \hline
    Notation & Value & Description\\\hline
    $A^{LoS}$ & 103.8 &\multirow{2}{*}{LoS path loss parameters}\\\cline{1-2}
    $B^{LoS}$ & 20.9 & \\\hline
    $A^{NLoS}$ & 145.4&\multirow{2}{*}{NLoS path loss parameters}\\\cline{1-2}
    $B^{NLoS}$ & 37.5&\\\hline
    $P_{tx}$ & 20 $dBm$ & Drone transmission power\\\hline
    $N_0$ & -104 $dBm$& Noise power\\\hline
    $h$ & 50 $m$& Drone altitude\\\hline
    $\alpha$ & 15 $dB$& SNR threshold\\\hline
    $\beta$ & 0 - $3g_{max}$ & Inner drone distance\\\hline
    $\gamma$ & 90\% - 98\% & Level of UE coverage \\\hline
    $s$ & 4 - 8$m/s$ & Flying speed of drones\\\hline
    $W$ & 100 $MHz$ & Bandwidth\\\hline
    $W_{max}$ & 2 $MHz$ & Upper bound of bandwidth at a UE\\\hline
    $h_R$ & 10 $m$ & Height of utility pole\\\hline
	$\lambda_S$ & 45\% & Percentage of one time slot for serving\\\hline
	$\lambda_F$& 5\% & Percentage of one time slot for flying\\\hline
	$\lambda_R$& 50\% & Percentage of one time slot for recharging\\\hline
    \end{tabular}
\end{center}
\end{table}

\begin{figure}[t]
    \centering
    \begin{subfigure}[b]{0.35\textwidth}
        \includegraphics[width=\textwidth]{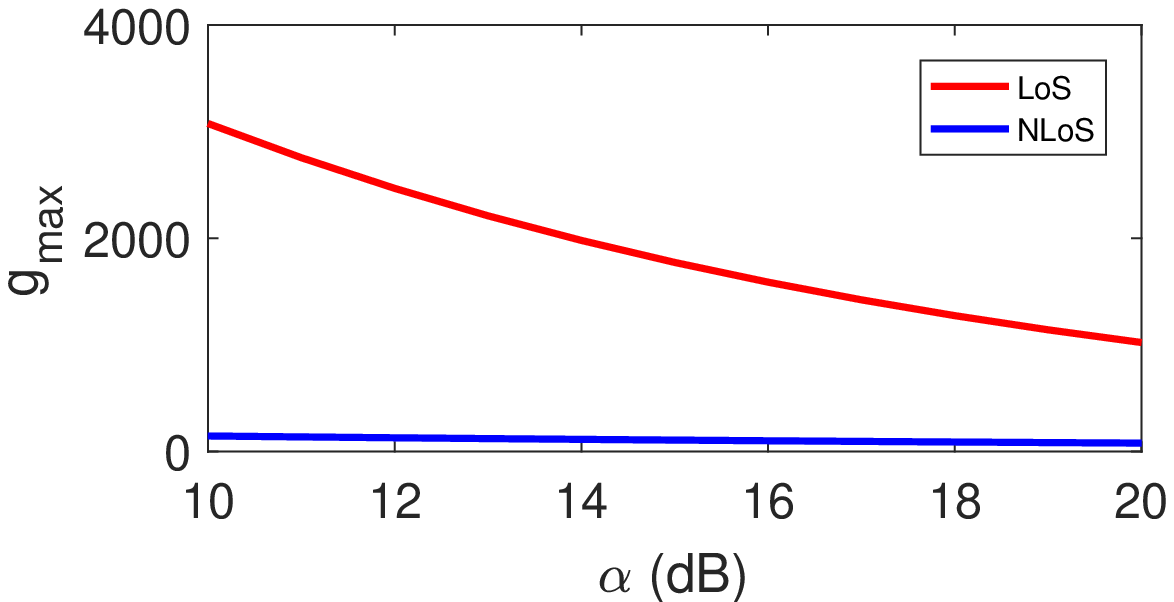}
        \caption{}
        \label{G_alpha}
    \end{subfigure}
    \begin{subfigure}[b]{0.35\textwidth}
        \includegraphics[width=\textwidth]{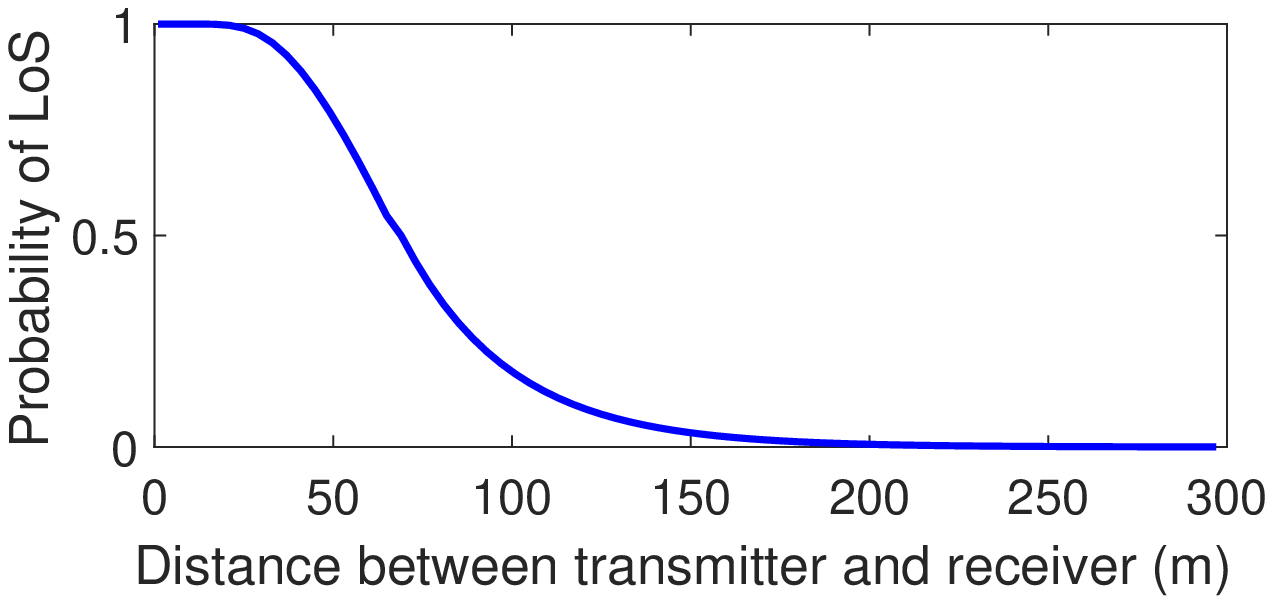}
        \caption{}
        \label{p_los}
    \end{subfigure}
    \begin{subfigure}[b]{0.35\textwidth}
        \includegraphics[width=\textwidth]{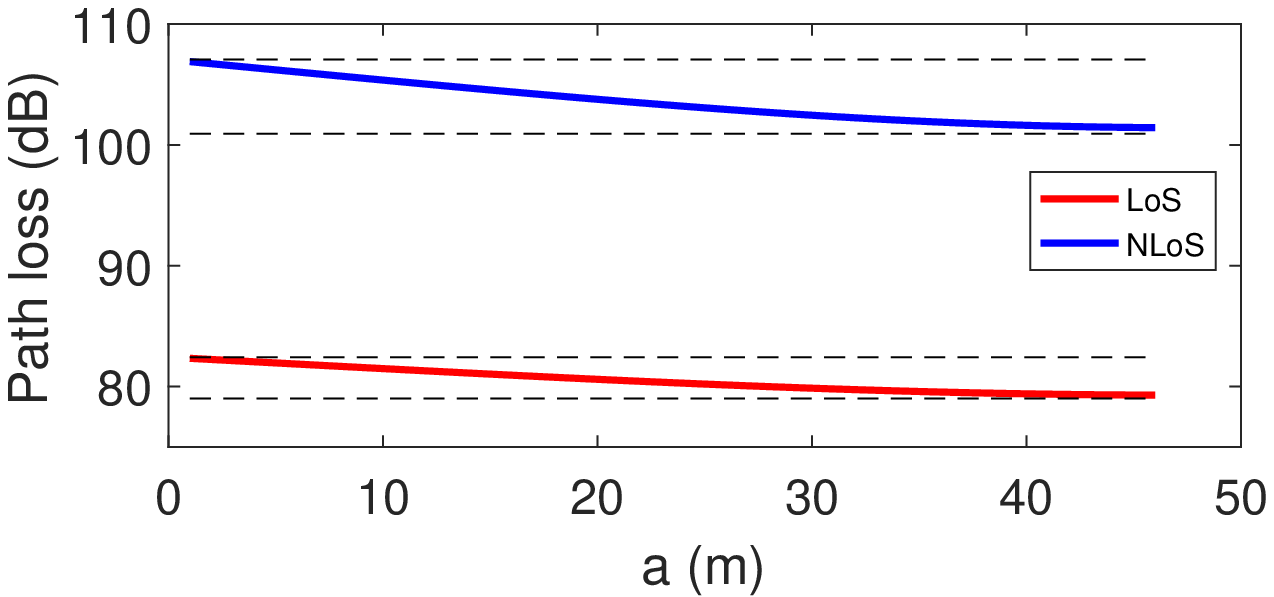}
        \caption{}
        \label{p_los_bound}
    \end{subfigure}
    \caption{(a) $g_{max}$ vs $\alpha$. (b) Probability of LoS vs distance. (c) Path loss.}
\end{figure}

%

Now we discuss how accurate of using graph distance to calculate path loss, instead of the physical distance. Consider two street segments $a$ and $b$ and they are perpendicular. The graph distance between the two end points is $a+b$ while the physical distance is $\sqrt{a^2+b^2}$. Let $g_{max}=a+b$ and we compute the path loss corresponding to the physical distance. Fig. \ref{p_los_bound} shows the path loss of both LoS and NLoS when $a$ is between 1 and $g_{max}/2$ meters. The physical distance varies with $a$, while the graph distance is fixed to $g_{max}$. Obviously, the largest gap occurs when $a=g_{max}/2$, and the corresponding path loss of LoS is only about 3 $dB$ lower, and that of NLoS is only about 6 $dB$ lower than the value calculated by the graph distance $g_{max}$, both of which incur small errors in practice. 
In short, approximating the physical distance by the graph distance has a small impact on the calculation of path loss. 


\subsection{Dataset}\label{dataset}
To get the UE density function, we make use of the dataset of Momo. When a Momo user has an update, the information of his ID, timestamp, latitude and longitude is sent to the server. The Momo dataset contains approximately 150 million such updates in a period of 38 days, from 21/5/2012 to 27/6/2012 \cite{chen2013and}. We extract a subset of this dataset, based on which we build up the UE density function. 

The Momo dataset consists of the updates of world-wide users. To make the selected dataset suitable to our problem, we only focus on the updates by the users in a small residential community in Beijing, China. The latitude of this area is from 39.9176N to 39.9242N and the longitude is from 116.4406E to 116.4501E, which is about 1059 $\times$ 721 $m^2$, as shown in Fig. \ref{map}. Further, we build up a discrete street graph according to the area map as shown in Fig. \ref{user_distribution_521}. Each street is represented by a number of discrete points. 
From the whole dataset, we select the updates whose locations fall into the considered latitude and longitude range. Based on our observation, these updates belong to indoor and outdoor UEs. We remove the indoor updates by selecting those falling into the neighbourhoods of the street points. Then, the selected dataset only consists of the updates belonging to the UEs near the streets. We demonstrate the distribution of such UEs on 21/5/2012 in Fig. \ref{user_distribution_521}.

\begin{figure}[t]
    \centering
    \begin{subfigure}[b]{0.47\textwidth}
        \includegraphics[width=\textwidth]{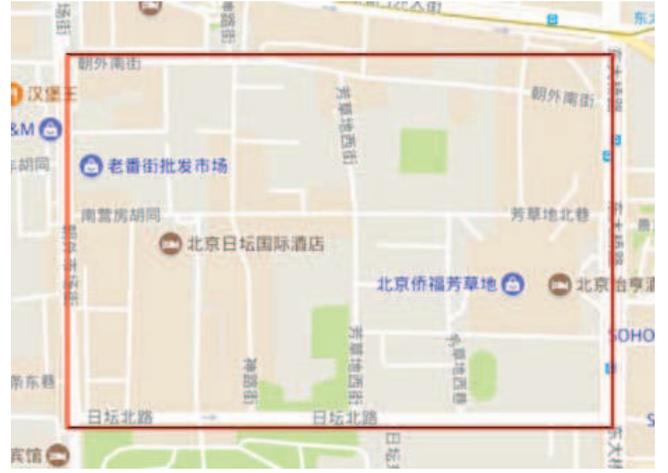}
        \caption{}
        \label{map}
    \end{subfigure}
    \begin{subfigure}[b]{0.53\textwidth}
        \includegraphics[width=\textwidth]{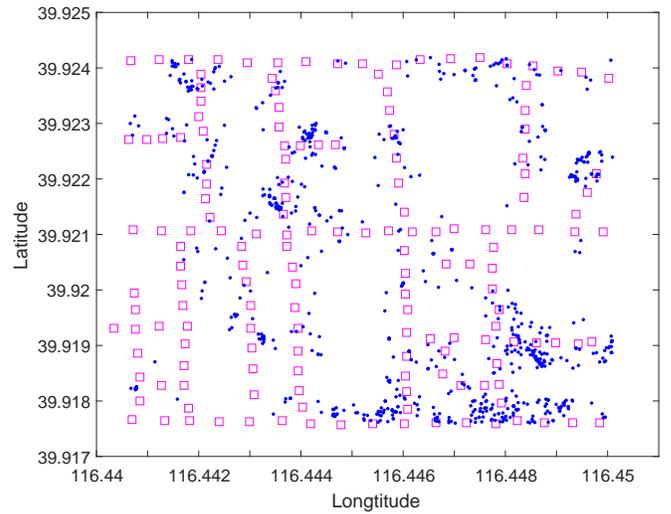}
        \caption{}
        \label{user_distribution_521}
    \end{subfigure}
    \caption{(a) The considered residential community in Beijing. (b) UE distribution on the street graph on 21/5/2012.}
\end{figure}


The selected subset of data covers more than five weeks including three types of UE pattens on weekdays, weekends, and public holidays. Since there are only three days belonging to category of public holidays, in this paper we only consider the UE pattens on weekdays and weekends. We respectively take the average numbers of UEs for weekdays and weekends and show them in Fig. \ref{UEnumber}, where the duration of a time slot is one hour. 
It can be seen that on weekdays the average UE number is mostly larger than weekends. The average total UE number on weekdays is 1003 per day while that on weekends is 727. On weekdays, the UE number increases steadily from 8:00 and arrives at the peak (about 70 UEs) at 16:00, after which it decreases. The UE variation for weekends is different: it increases slowly from 8:00 to 16:00, and from 16:00 to 23:00 the UE number remains at around 40. After 23:00, much later than weekdays, it drops down. If we set 40 as a UE number threshold to decide the usage of drones: on weekdays we need to send drones to serve UEs from 10:00 to 22:59, i.e., 13 hours; while on weekends the drones should work for 8 hours between 15:00 and 22:59. We need to mention that the number of UEs provided by Fig. \ref{UEnumber} shows only the UEs using Momo. The actual number of UEs should also include those not using Momo. Although the dataset we have cannot provide us the true total number of UEs, it presents the realistic traffic pattens and UE distributions, which is much more closer to reality than random distribution. To make the results more sensible, we introduce a scalar to scale up the UE number according to \cite{scale15}, which is set as 5 in the simulations.


\begin{figure}[t]
\begin{center}
{\includegraphics[width=0.53\textwidth]{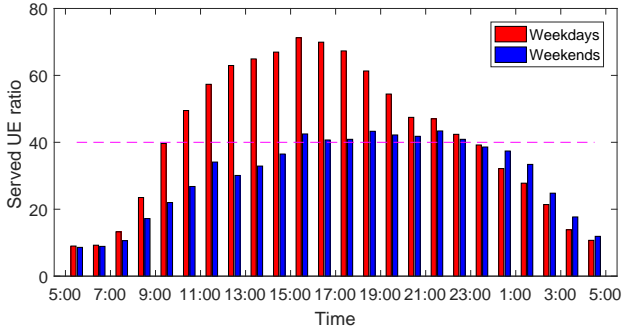}}
\caption{Average UE numbers using Momo on weekdays and weekends in the considered area.}\label{UEnumber}
\end{center}
\end{figure}

\subsection{Metrics}\label{metric}
In the evaluations, we consider the performance metrics:
\begin{itemize} 
\item Served UE ratio: The served UE ratio is defined as $|\bigcup_{v\in {\cal V}}U(v)|/|{\cal U}|\times 100\%$. 
\item Spectral efficiency ($SE$): a measure of how efficiently a limited frequency spectrum is utilized by PHY layer protocol. The spectral efficiency at a UE is computed by $SE=log_2(1+SINR)$, where SINR is defined in (\ref{SINR}). In this paper, we consider the average spectral efficiency:
\begin{equation}
ASE=\frac{1}{M} \sum_{i=1}^M log_2(1+SINR_i)
\end{equation} 
\item Number of drones: the number of drones required to serve at least $\gamma$ percent of UEs. 
\item Served UE per drone: the average number of UEs served by one drone. 
\item Network capacity ($NC$): the number of bits that can be generated in an unit area of 1 $km^2$, given the bandwidth $W$. The network capacity is a metric measuring the overall performance in both PHY and MAC layers, whose unit is $Mbps/km^2$. For simplicity, we assume that the bandwidth of a drone is evenly allocated to its served UEs and there is a limit on the bandwidth that each UE can get, i.e., $W_{max}$. As a result, $W_{i}=\min(\frac{W}{M_j}, W_{max})$, where $M_j$ is the number of UEs served by drone $j$. Suppose there are $k$ drones in 1 $km^2$, then the network capacity can be computed by:
\begin{equation}
NC=\sum_{j=1}^k \sum_{i=1}^{M_j} SE_i\cdot W_i
\end{equation}
\end{itemize}

\subsection{Comparing Approaches}\label{compared_approach}
As mentioned in Sections \ref{introduction} and \ref{review}, there are several related work on the topic of drone deployment problem. For the scenario of placing a single drone, the most relevant work is  \cite{alzenad20173d}. Aiming at serving the maximum UEs, the authors formulate a mixed integer non-linear problem to find the optimal position of the drone. Further, the authors formulate a second order cone problem to shorten the cover range of the drone to save transmission energy. 
For a fair comparison, we only compare our work with the results of the mixed integer non-linear problem (MINLP), since shortening a bit transmission range does not contribute much to the total energy consumption by the drone. We acknowledge that the energy consumed for drone movement is the dominant factor.

There are also some related work about placing multiple drones, such as \cite{sharma2016uav} and \cite{kalantari2016number}. However, both of them divide the area of interest into a set of zones or subareas, which are quite different from our basic model, i.e., the street graph. So we demonstrate the comparison of our approach with \cite{alzenad20173d} for the case of single drone deployment; and the comparison with max $k$-cover (without inner drone distance constraint) for the case of multiple drone deployment. 

%
%
%

\subsection{Simulation results}\label{results}
We present simulation results to evaluate the performance of the proposed solutions respectively in this part.
\subsubsection{Evaluation of SDD}

We first show the performance of single drone deployment method. As discussed in Section \ref{dataset}, on weekdays we deploy a drone from 10:00 to 22:59 for 13 hours; while on weekends the drone work for 8 hours between 15:00 and 22:59. Except \cite{alzenad20173d} (named as Approach 1), we also compare with a random deployment (named as Random). 

We respectively display the served UE ratios on weekdays and weekends in Fig. \ref{average_UE_weekdays_5} and \ref{average_UE_weekends}. We can see that our proposed method achieves similar performance compared with \cite{alzenad20173d} in terms of served UE ratio. In theory, the 2D projection of a drone is constrained on street in our approach; while there is no such limitation in \cite{alzenad20173d}. So the optimum solution by \cite{alzenad20173d} achieves no worse performance in served UE ratio than ours. However, \cite{alzenad20173d} uses MOSEK solver to address the MINLP and some solutions are local optimums, which are not competitive to our solutions. Further, both of them outperform the random deployment. 

It is worthy of mentioning that, as there is no limitation on the 2D projection of the drone, over 70\% of the solutions by \cite{alzenad20173d} locate off the streets (see the illustrative example shown in Fig. \ref{deployment_example}, although the deployment by \cite{alzenad20173d} serves more UEs, the projection is off street), which have high probability of hitting tall buildings if applied in urban environment directly. In contract, our solutions are on streets and can be applied directly to realistic networks, thanks to the street graph model.

\begin{figure}[t]
    \centering
    \begin{subfigure}[b]{0.47\textwidth}
        \includegraphics[width=\textwidth]{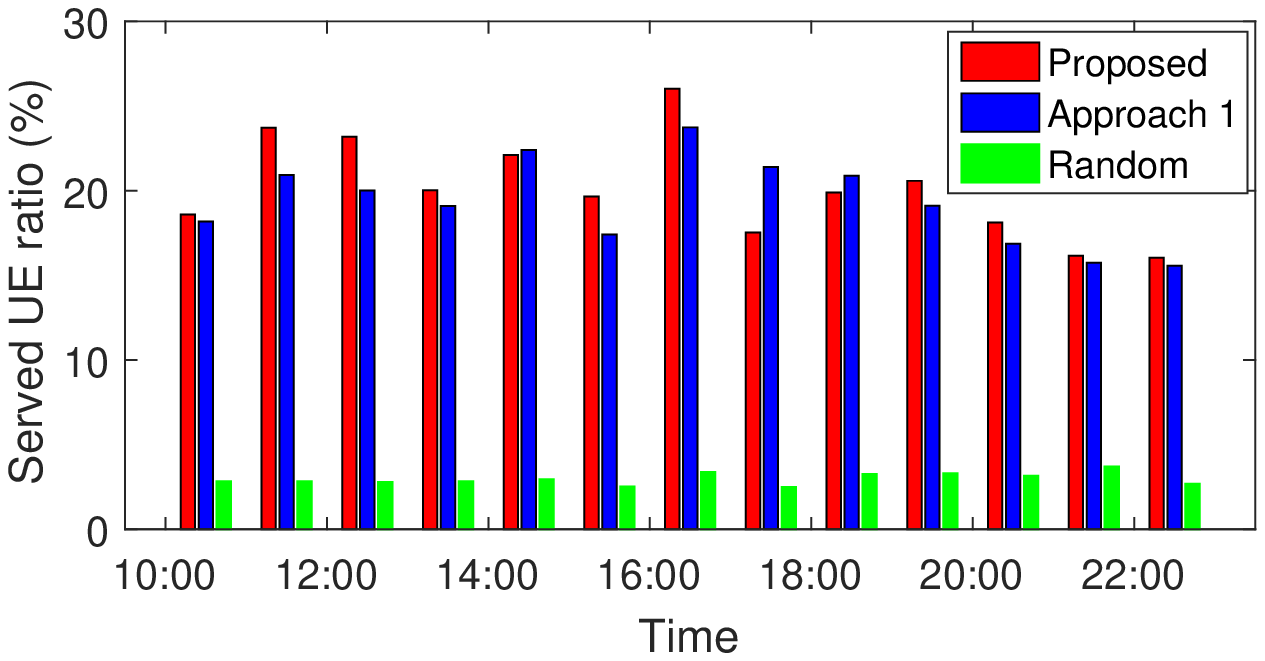}
        \caption{}
        \label{average_UE_weekdays_5}
    \end{subfigure}
    \begin{subfigure}[b]{0.47\textwidth}
        \includegraphics[width=\textwidth]{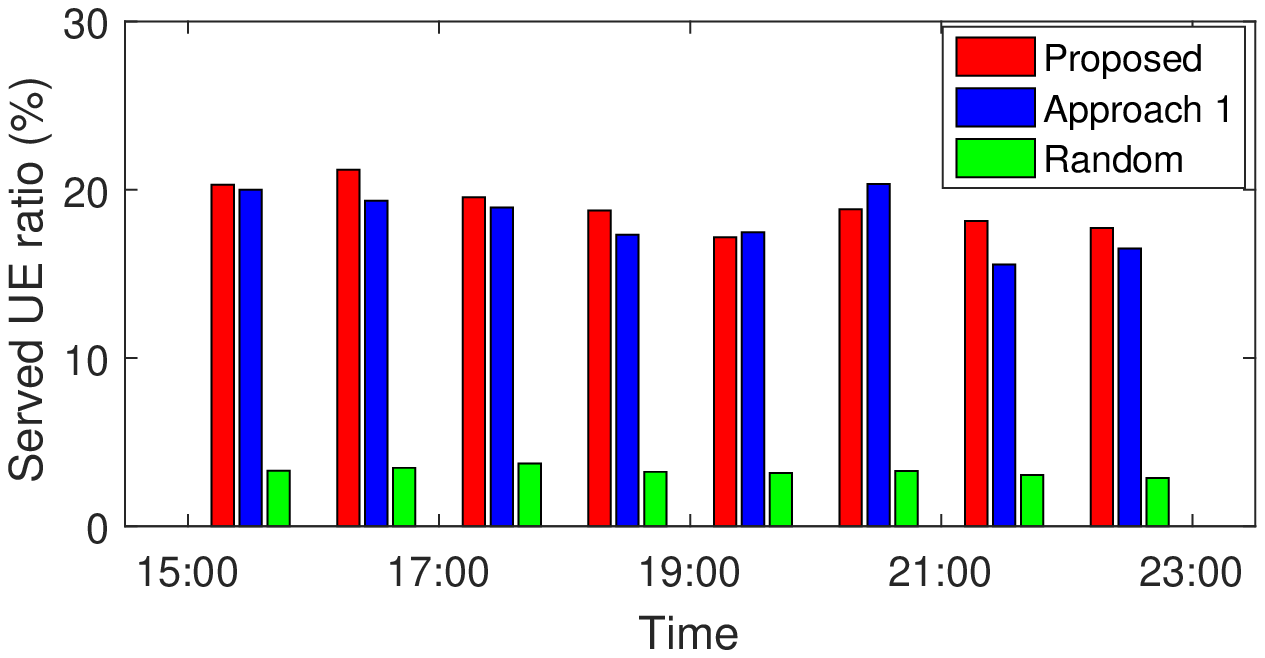}
        \caption{}
        \label{average_UE_weekends}
    \end{subfigure}
    \caption{(a) Average served UE ratio on weekdays. (b) Average served UE ratio on weekends.}
\end{figure}

%

\begin{figure}[t]
\begin{center}
{\includegraphics[width=0.45\textwidth]{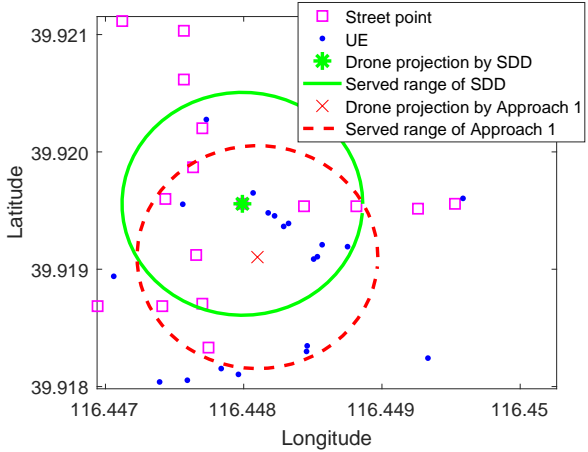}}
\caption{Illustrative example of 2D projections by the proposed approach and Approach 1.}\label{deployment_example}
\end{center}
\end{figure}

\subsubsection{Evaluation of kDD}

\begin{figure}[t]
    \centering
    \begin{subfigure}[b]{0.47\textwidth}
        \includegraphics[width=\textwidth]{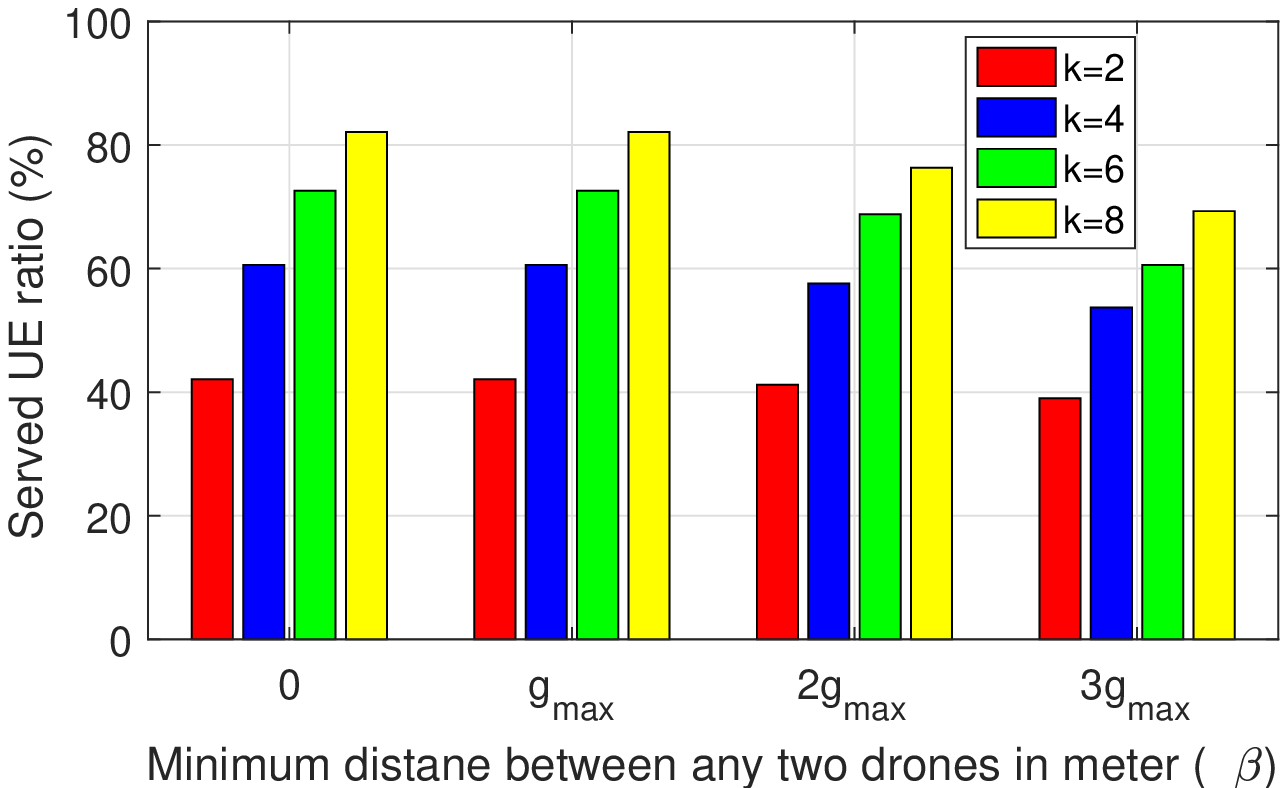}
        \caption{}
        \label{UE_ratio_beta_k_17}
    \end{subfigure}
    \begin{subfigure}[b]{0.47\textwidth}
        \includegraphics[width=\textwidth]{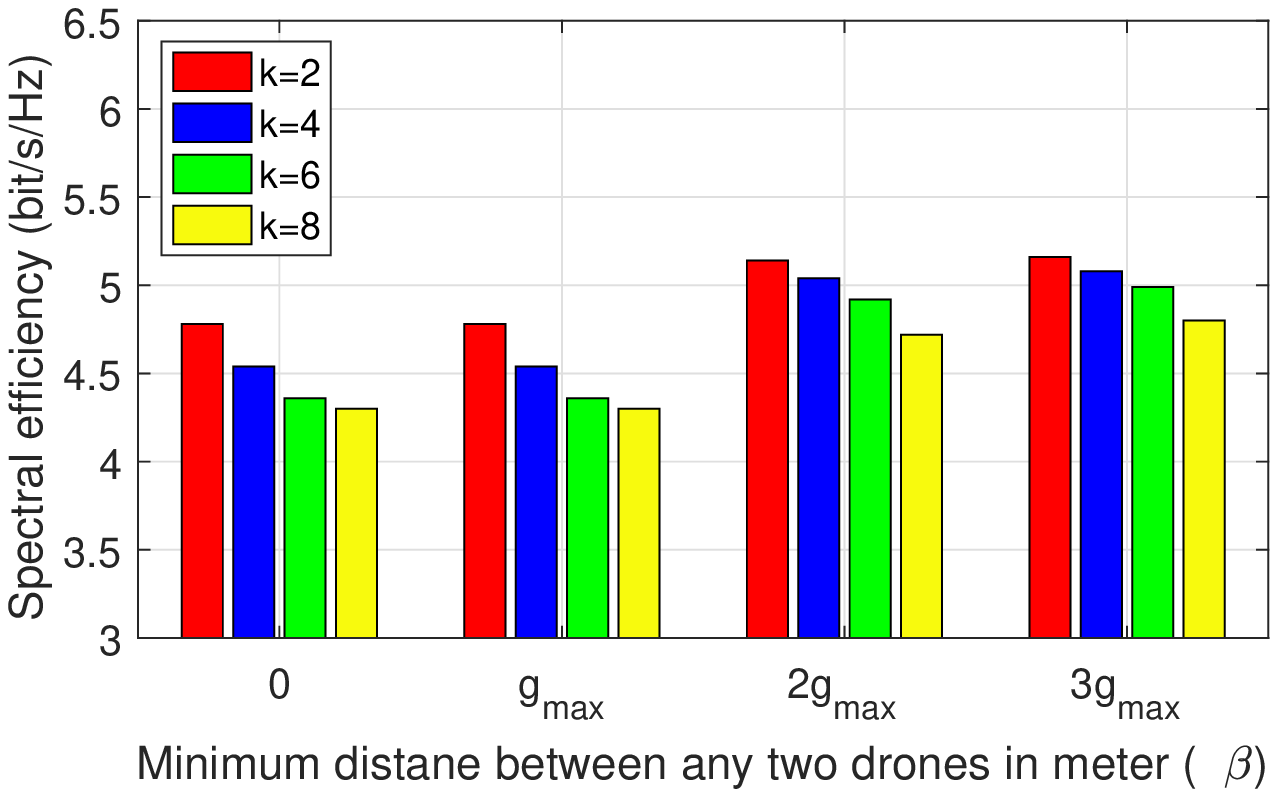}
        \caption{}
        \label{rate_beta_k_17}
    \end{subfigure}
    \caption{(a) Average served UE ratio on weekdays by multiple drones. The average number of UEs during the peak hour on weekdays is 350. (b) Average spectral efficiency on weekdays by multiple drones.}
\end{figure}

%

We demonstrate simulation results for deploying multiple drones. When multiple drones are used, as mentioned before, the interference should be taken into account to assess the user experience. Two main factors influence the interference intensity: the number of drones $k$, and the distance threshold between two drones $\beta$. Here, we vary $k$ from 2 to 8 and $\beta$ from $0$ to $3g_{max}$. Note that, when $\beta$ is 0, we are solving the max $k$-cover problem, i.e., without the consideration of inner drone distance constraint, because constraint (\ref{constraint4_3}) is always satisfied.

To have an insight of the impacts of these two factors, we focus on the performance on the peak hour, i.e., from 15:00 to 15:59, on weekdays.
We present the served UE ratio and spectral efficiency in Fig. \ref{UE_ratio_beta_k_17} and \ref{rate_beta_k_17} for the peak hour. Note, the results shown here are the average values of 27 weekdays. From these results we can see that: 
\begin{itemize}
\item With the increase of $k$, the served UE ratio increases; while the average spectral efficiency decreases. Because more UEs are covered by drones when $k$ is increased, and in the same time, the interference is also  raised.
\item With the increase of $\beta$, the served UE ratio slightly decreases; while the average spectral efficiency increases. This is because increasing $\beta$ means increasing the minimum distance between any two drones, which reduces the coverage performance; but in the meanwhile, it also reduces the strength of interference at a UE from other drones, which leads to larger spectral efficiency.
\item An interesting finding is that the results when $\beta$ takes $g_{max}$ are very similar to those when $\beta=0$ in terms of both served UE ratio and average spectral efficiency. In other words, the standard greedy algorithm for max $k$-cover problem inherently avoids placing two drones closely to each other.
\end{itemize}

\begin{figure}[t]
    \centering
    \begin{subfigure}[b]{0.47\textwidth}
        \includegraphics[width=\textwidth]{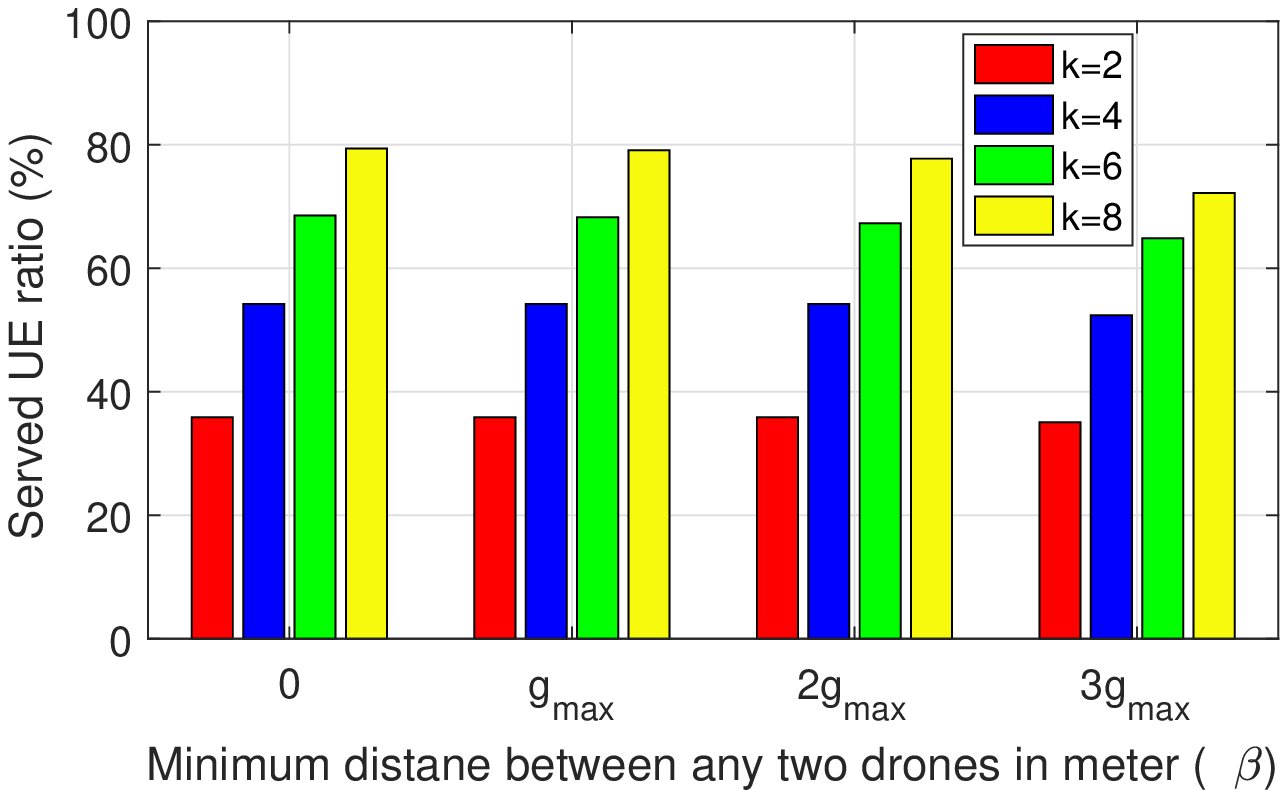}
        \caption{}
        \label{UE_ratio_beta_k_17_weekends}
    \end{subfigure}
    \begin{subfigure}[b]{0.47\textwidth}
        \includegraphics[width=\textwidth]{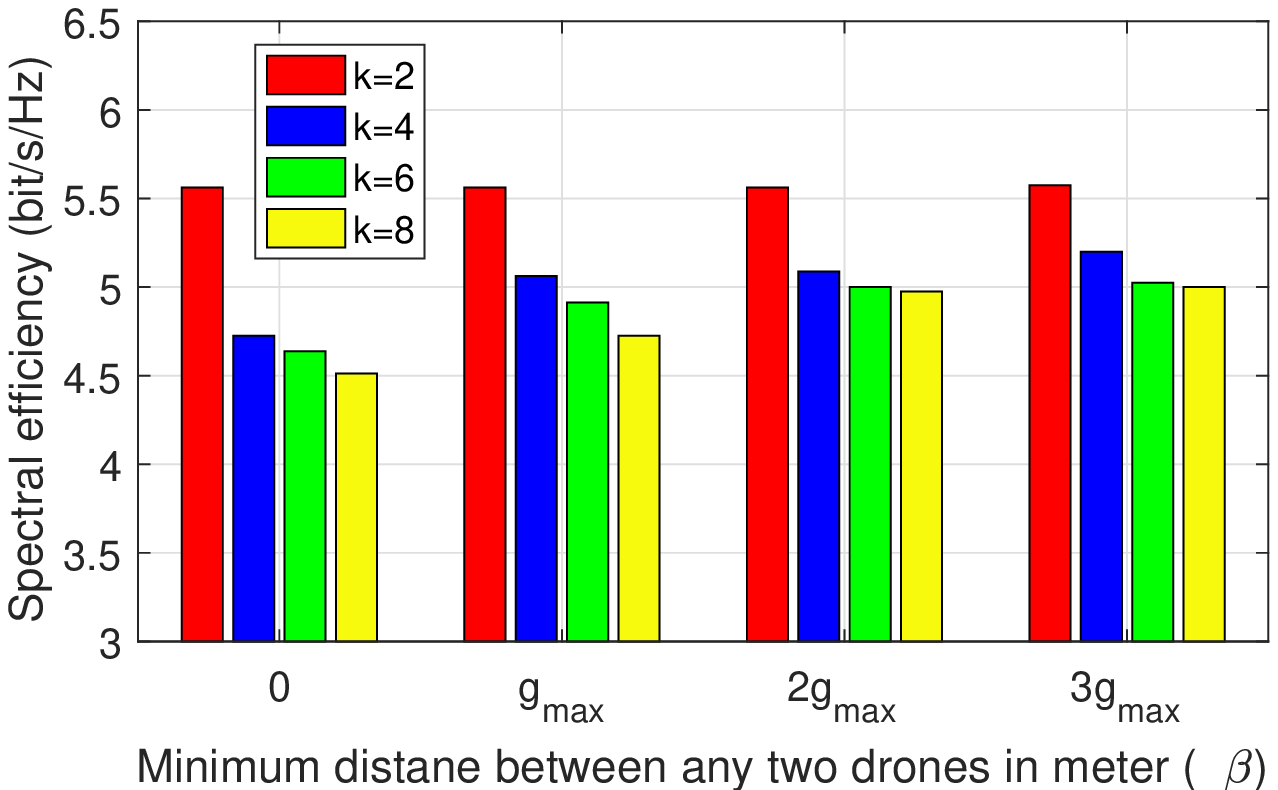}
        \caption{}
        \label{rate_beta_k_17_weekends}
    \end{subfigure}
    \caption{(a) Average served UE ratio on weekends by multiple drones. The average number of UEs during the peak hour on weekends is 210. (b) Average spectral efficiency on weekends by multiple drones.}
\end{figure}

%

The results for the same peak hour on the weekends are shown in Fig. \ref{UE_ratio_beta_k_17_weekends} and \ref{rate_beta_k_17_weekends}. The performance trends are similar to those of weekdays. Comparing the results of weekdays and weekends, we can find that under the same $\beta$ and $k$, the served UE ratio on weekdays is similar to weekends (see Fig. \ref{UE_ratio_beta_k_17} and \ref{UE_ratio_beta_k_17_weekends}). But the number of served UEs on weekdays is much larger than weekends due to the dense distribution of UEs on weekdays. The spectral efficiency on weekdays is about 10\% lower than that on weekends in average (see Fig. \ref{rate_beta_k_17} and \ref{rate_beta_k_17_weekends}). The reason behind this phenomenon is that the distribution of UEs on weekends is more sparse than weekdays. Under the same $\beta$ and $k$, the deployments of drones is also more sparse, which leads to lower interference, then higher spectral efficiency.

\subsubsection{Evaluation of EkDD}
In this part, we demonstrate the performance of the solution to the problem of EkDD. 

We assume that the operation power and the energy recharging rate are the same, i.e., $p=q$. From (\ref{n_k}) we have $n\geq \frac{1}{2}k$, i.e., the drones should be  divided into 2 groups. Consider the time slot of 1 hour in this paper and the state-of-the-art commercial drones, such as DJI\footnote{https://www.dji.com}, whose the flying time are around 30 minutes. So in this paper, we set $\lambda_S$, $\lambda_F$, and $\lambda_R$ to 45\%, 5\% and 50\% respectively. We consider that there are four utility poles located at the corners of graph, see Fig. \ref{snapshot} for an example. Here, we fix $\beta$ as $g_{max}$.

As discussed in Section \ref{P4}, the flying speed is the main factor impacting on the drone deployment. Here, we consider various practical flying speeds $4,5,6,7,8 m/s$, and present the corresponding performance in served UE ratio for the considered peak hour in Fig. \ref{P4_served_UE_ratio}. The results of kDP are also displayed for comparison, indicated by 'Inf'. For a fair comparison, in EkDD, $k=8$; while in kDD, $k=4$. Then, in both the cases, 4 drones can serve UEs simultaneously. From Fig. \ref{P4_served_UE_ratio} we can see that with the increase of $s$, the served UE ratio increases first and then remain at a steady level. Because a larger $s$ means wider operation radius for drones, and thus a weaker constraint on the positions of drones. When $s$ is larger than $6m/s$, the recharging constraint becomes invalid. Fig. \ref{snapshot} shows the drone projections for four cases of $s=4m/s$, $s=5m/s$, $s=6m/s$ and 'Inf', from which we can find that the constraint (\ref{constraint4_4}) pulls the drones' positions closer to the recharging positions compared to the non-constraint case of 'Inf'. When $s$ takes 7 or $8m/s$, the positions of drones are the same with the cases of $6m/s$ and 'Inf'. So we do not display them in Fig. \ref{snapshot}.

\begin{figure}[t]
    \centering
    \begin{subfigure}[b]{0.47\textwidth}
        \includegraphics[width=\textwidth]{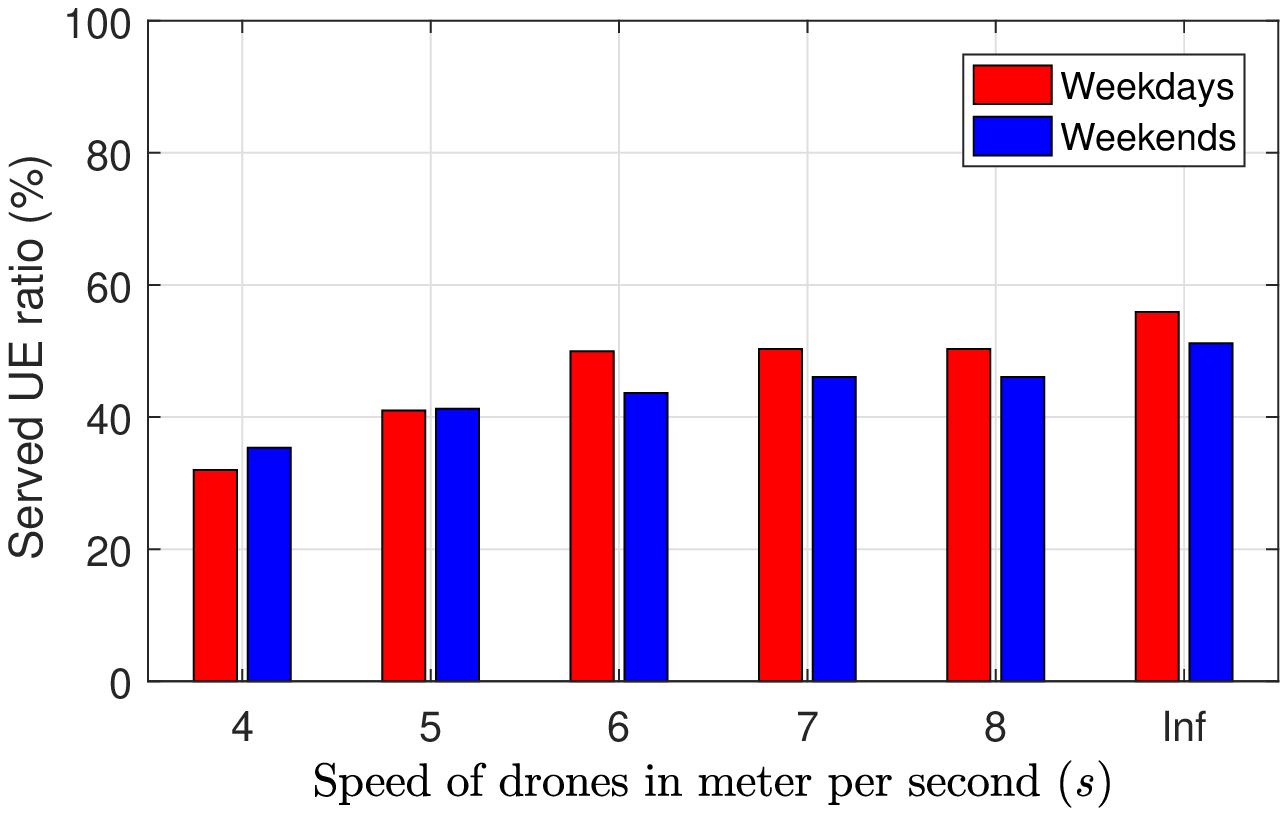}
        \caption{}
        \label{P4_served_UE_ratio}
    \end{subfigure}
    \begin{subfigure}[b]{0.47\textwidth}
        \includegraphics[width=\textwidth]{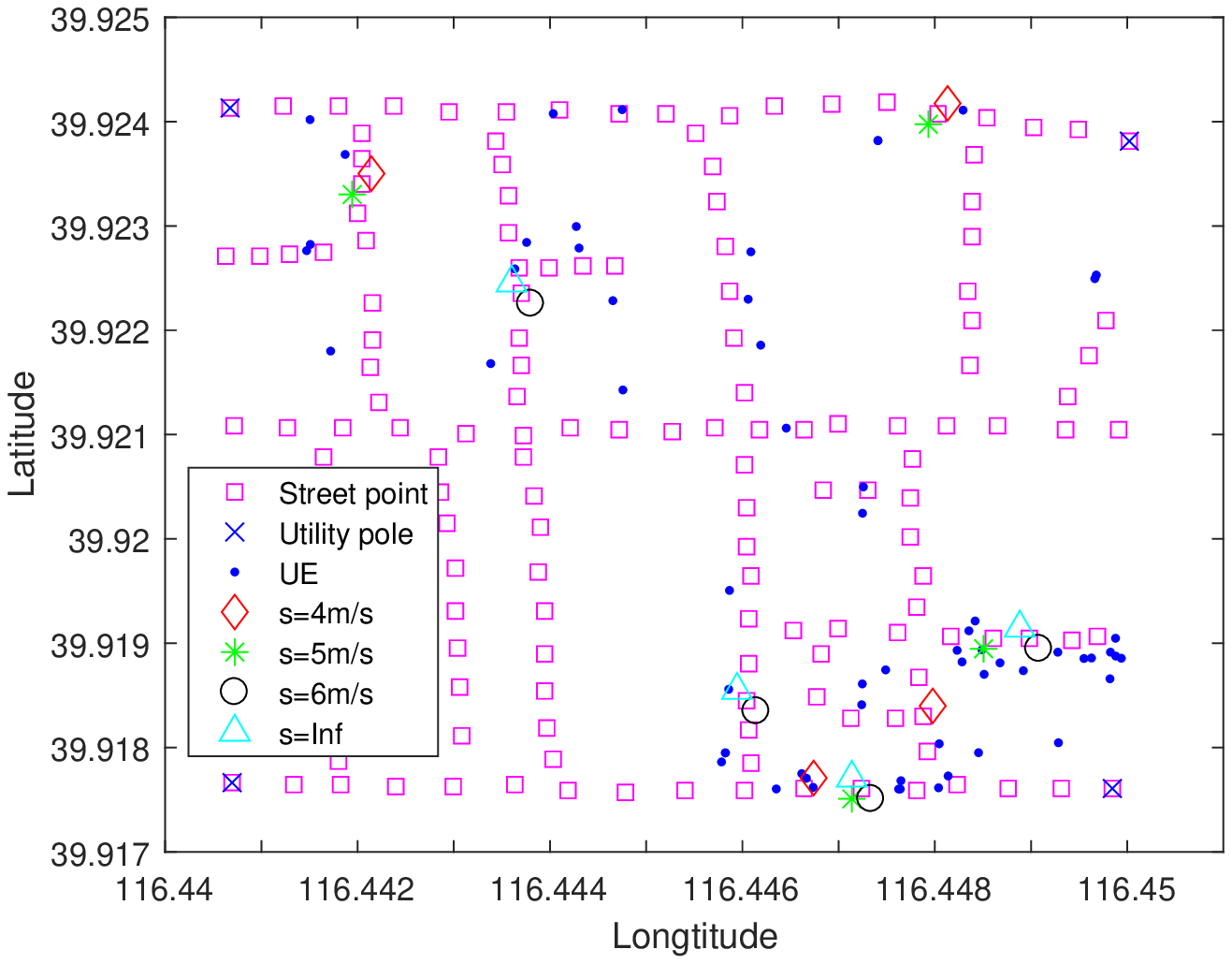}
        \caption{}
        \label{snapshot}
    \end{subfigure}
    \caption{(a) Served UE ratio against flying speed. The total numbers of UEs are 350 and 210 for the peak hour on weekdays and weekends respectively. (b) Drone projections in the cases of $4m/s$, $5m/s$, $6m/s$, and 'Inf' for the peak hour on a weekday.}
\end{figure}

%

In short, taking into account the energy issue, the performance reflects a more practical scenario. The drones with higher flying speed can achieve better performance due to the relaxed constraint on its position to recharge its battery.

\subsubsection{Evaluation of MinDD}
In this part, we show how many drones are needed to serve at least $\gamma$ percent of UEs. Here, we fix $\beta$ as $g_{max}$ and show the relationship between $\gamma$ and the minimum number of required drones, number of UEs served per drone, and network capacity. We vary $\gamma$ from 90\% to 98\%. The results are shown in Fig. \ref{evaluation4}.

\begin{figure}[t]
    \centering
    \begin{subfigure}[b]{0.47\textwidth}
        \includegraphics[width=\textwidth]{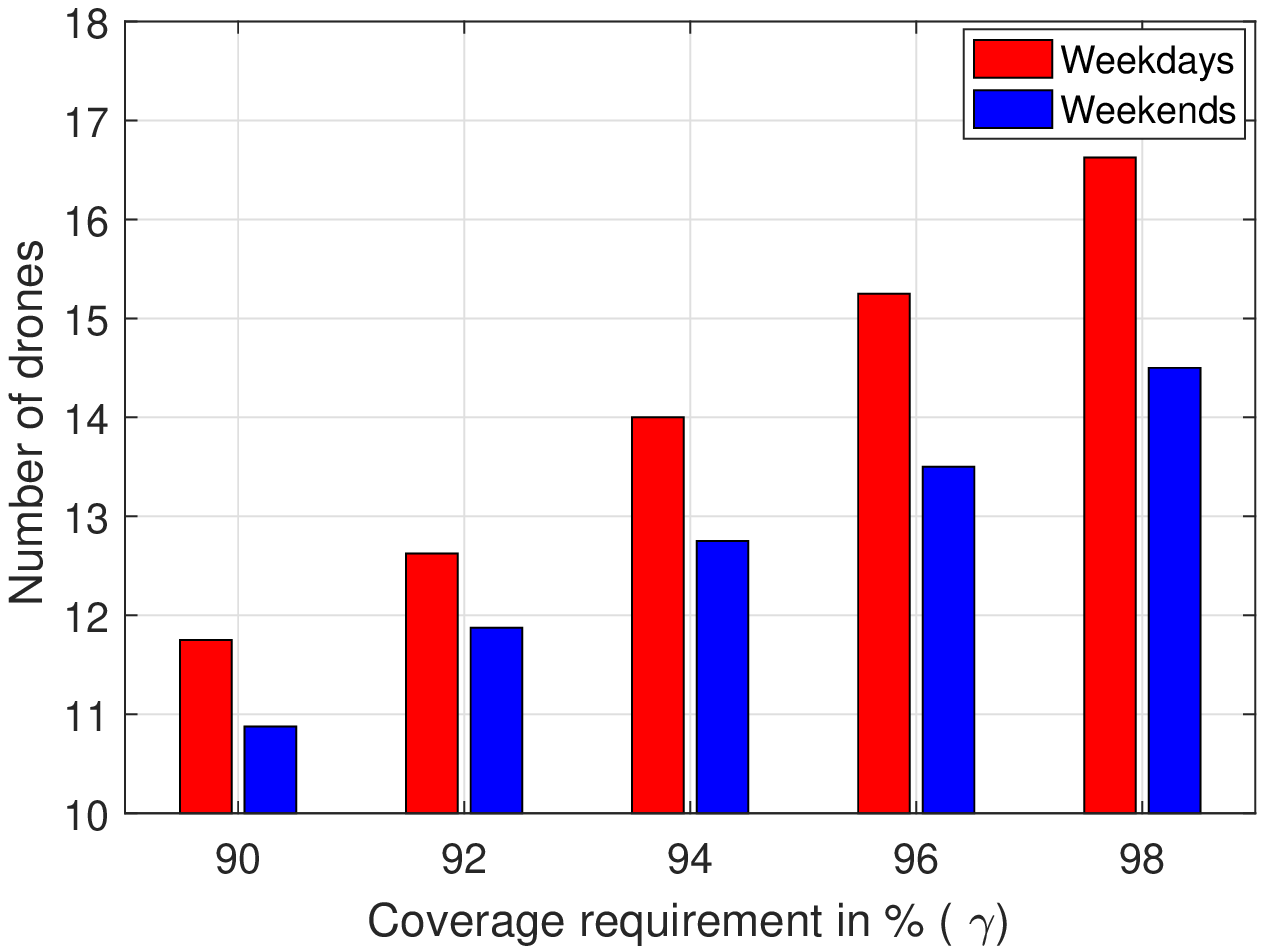}
        \caption{}
        \label{droneNo_gamma_17}
    \end{subfigure}
    \begin{subfigure}[b]{0.47\textwidth}
        \includegraphics[width=\textwidth]{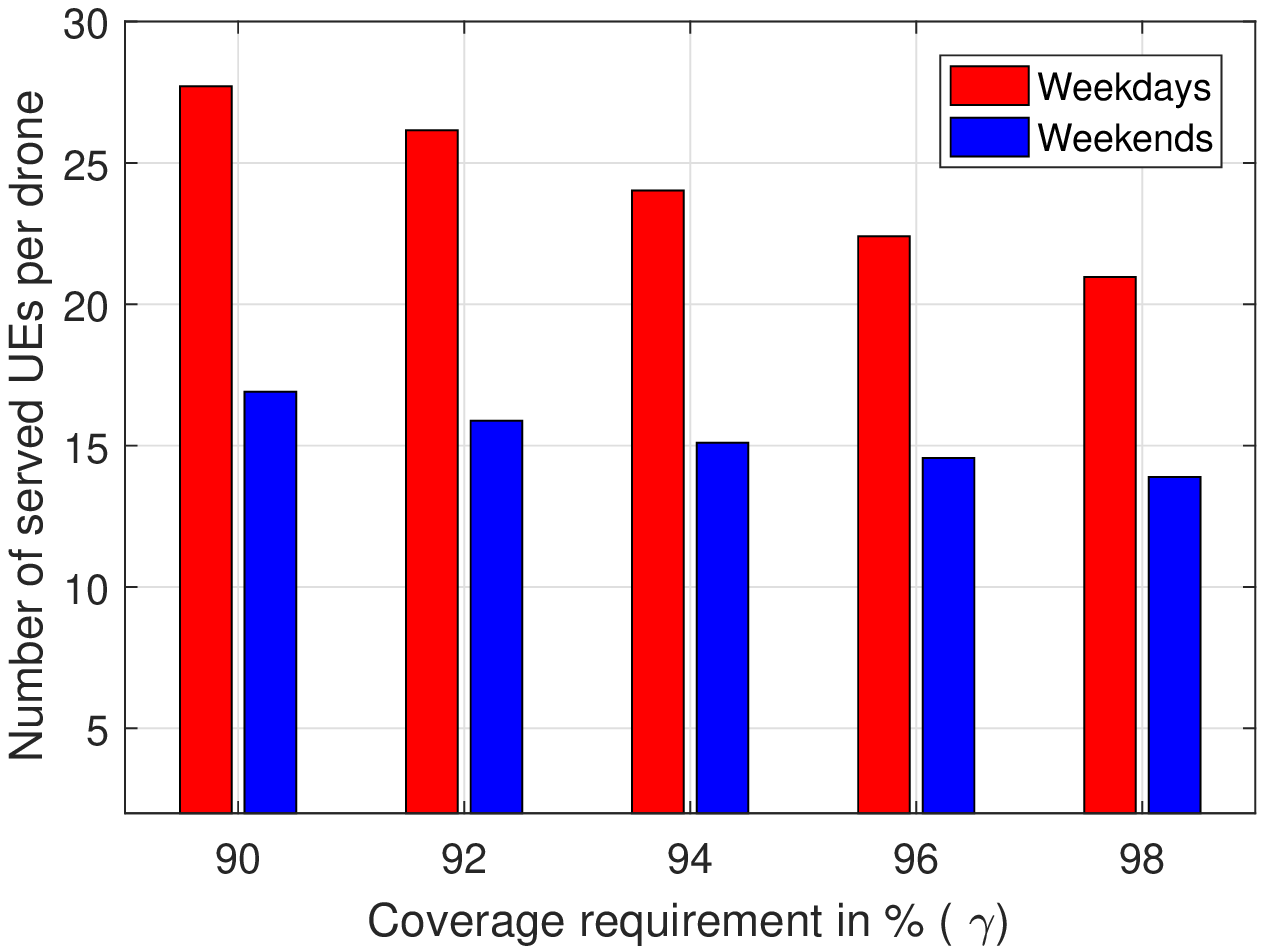}
        \caption{}
        \label{UE_per_drone_gamma_17}
    \end{subfigure}
    \begin{subfigure}[b]{0.47\textwidth}
        \includegraphics[width=\textwidth]{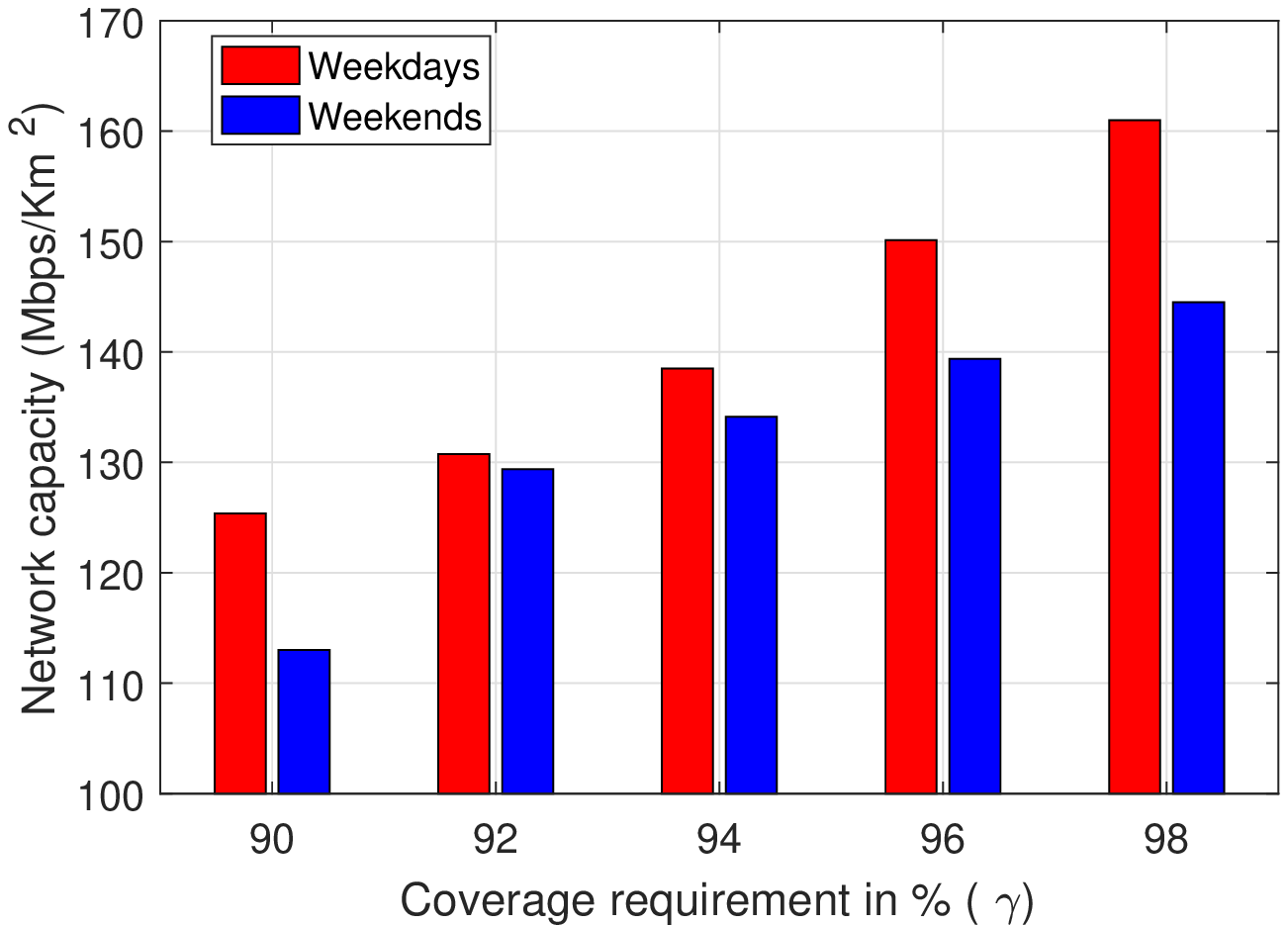}
        \caption{}
        \label{network_capacity_17}
    \end{subfigure}
    \caption{(a) Average minimum number of drones to serve $\gamma$ percent of UEs. (b) Average number of served UEs per drone against $\gamma$. (c) Network capacity against $\gamma$.}\label{evaluation4}
\end{figure}

%

\begin{itemize}
\item Fig. \ref{droneNo_gamma_17} displays the required drone numbers with various $\gamma$ on weekdays and weekends for the considered peak hour. We can see that the minimum number of drones increases with $\gamma$. 
Since the number of UEs on weekends are much smaller than weekdays, fewer drones are needed on weekends than on weekdays to achieve the same level of UE coverage. 
\item Fig. \ref{UE_per_drone_gamma_17} shows that, with the increase of $\gamma$, the average number of UEs served per drone decreases. In other words, the increase of served UE number is slow than the increase of drone number. We can also see that on weekdays the number of UEs served per drone is much higher than weekends, because UEs are sparsely distributed on weekends. 


\item Fig. \ref{network_capacity_17} presents the network capacity against $\gamma$. We can see that the network capacities on weekdays and weekends both increase with $\gamma$. Since the total number of UEs on weekends is much smaller than that on weekdays, hence the resource at drones is not fully used. So the network capacity on weekdays is larger than that on weekends.
\end{itemize}

From the above simulation results we can find that to serve more UEs, more drones are required and the network capacity can be improved. However, the efficiency of drone usage, i.e., the average number of served UEs per drone, decreases. Therefore, from the view of ISPs, there is a trade-off between the investment, i.e., the number of drones, and performance the gain, i.e., the network capacity.

\section{Limitations and Concluding Remarks}\label{conclusion}

In this paper, we adopt a UE density model based on the collected dataset Momo. However, the modern wireless traffic demand is quite dynamic. Thus, more effective models to predict the time-variant UE density are for further study. Besides, we have not touched the performance impact of drone altitude, which is another significant factor in practical drone deployment. Although it is a regulation issue, finding the optimal altitude within the allowed range may improve the number of served UEs. Further, the objective function of EkDD is a rough estimation of the served UE number, because some drones can fly back to the serving positions using less time than $\lambda_F$ of a time slot. The precise formulation requires to consider the serving positions in the objective function, which makes the problem more complex. Moreover, in terms of radio resource management (RRM), we simply allocate the available bandwidth evenly to UEs. Optimal allocation of radio resource is worth investigating to achieve a larger network capacity \cite{abdelnasser2016resource, jafari2015study}. 

From simple to complex, in this paper we discussed how to deploy a given number of drones to maximize the effectively served UE number; and how to determine the minimum number of drones to achieve a certain level of UE coverage. Different from existing work, the drone deployment is formulated based on a street graph model in this paper. The street graph, associated with UE density function (obtained from a realistic dataset Momo), is close to reality. Further, we proved that the problems are NP-hard, and then greedy algorithms were proposed to solve the problems. The effectiveness of our approaches was verified by extensive simulations on the Momo dataset. Since both the street graph model and the UE density function are quite realistic, the results provided in this paper provide valuable guidelines for realistic applications.


\bibliographystyle{IEEEtran}
\bibliography{mybibtex}
\end{document}